\documentclass{siamart171218}

\usepackage{lipsum}
\usepackage{amsfonts}
\usepackage{graphicx}
\usepackage{epstopdf}

\usepackage{booktabs} % For formal tables

\usepackage{algorithm}
\usepackage[noend]{algorithmic} % For algorithms

\usepackage{amssymb}
\usepackage{enumerate}

\usepackage{stmaryrd}
\usepackage{mathtools}
\DeclarePairedDelimiter\ceil{\lceil}{\rceil}
\DeclarePairedDelimiter\floor{\lfloor}{\rfloor}
\DeclarePairedDelimiter\bbracket{\llbracket}{\rrbracket}

\newcommand{\agree}{{\rm agree}}

\newcommand{\F}{\mathbb{F}}

\newcommand{\bA}{\mathbb A}

\newcommand{\bE}{\mathbb E}

\newcommand{\bI}{\mathbb I}
\newcommand{\bM}{\mathbb M}
\newcommand{\bP}{\mathbb P}
\newcommand{\bW}{\mathbb W}
\newcommand{\bV}{\mathbb V}

\DeclareMathAlphabet{\mathbfsl}{OT1}{ppl}{b}{it} %{OT1}{cmr}{bx}{it}

\newcommand{\ba}{\mathbfsl{a}}
\newcommand{\bb}{\mathbfsl{b}}
\newcommand{\bc}{\mathbfsl{c}}

\newcommand{\bs}{\mathbfsl{s}}
\newcommand{\bt}{\mathbfsl{t}}
\newcommand{\bu}{\mathbfsl{u}}
\newcommand{\bv}{\mathbfsl{v}}
\newcommand{\bw}{\mathbfsl{w}}
\newcommand{\bp}{\mathbfsl{p}}

\newcommand{\bsg}{\pmb{\sigma}}

\newcommand{\bcc}{\mathbfsl{C}}
\newcommand{\bcs}{\mathbfsl{S}}
\newcommand{\bct}{\mathbfsl{T}}

\newcommand{\bx}{\mathbfsl{x}}

\newcommand{\bai}{\bar i}
\newcommand{\baj}{\bar j}
\newcommand{\bam}{\bar m}

\newcommand{\hai}{\hat i}
\newcommand{\haj}{\hat j}

\newcommand{\cC}{{\mathcal C}}
\newcommand{\cG}{{\mathcal G}}

\newcommand{\cM}{{\mathcal M}}

\newcommand{\hav}{\hat \bv}
\newcommand{\bav}{\bar \bc}

\newcommand{\red}{{\rm red}}

\newcommand{\Crs}{{\mathcal C}_{\rm RS}}
\newcommand{\Diag}{{\rm Diag}}
\newcommand{\enc}{{\rm enc}_{\rm RS}}
\newcommand{\dec}{{\rm dec}_{\rm RS}}

\newcommand{\decG}{{\rm dec}_{\rm Gray}}

\newcommand{\ppmod}[1]{~({\rm mod~}#1)}

\ifpdf
  \DeclareGraphicsExtensions{.eps,.pdf,.png,.jpg}
\else
  \DeclareGraphicsExtensions{.eps}
\fi

\newsiamremark{example}{Example}
\newsiamremark{remark}{Remark}
\newsiamremark{hypothesis}{Hypothesis}
\crefname{hypothesis}{Hypothesis}{Hypotheses}
\newsiamthm{claim}{Claim}

\headers{Robust Positioning Patterns}{Y. M. Chee, D. T. Dao, H. M. Kiah, S. Ling, and H. Wei}

\title{Robust Positioning Patterns with Low Redundancy\thanks{Part of results in this paper were presented at SODA 2019.}}

\author{Yeow Meng Chee\thanks{Department of Industrial Systems Engineering and Management, National University of Singapore, Singapore
  (\email{pvocym@nus.edu.sg}).  This work was done while the author was with School of Physical and Mathematical Sciences,  Nanyang Technological University, Singapore. }
\and Duc Tu Dao\thanks{School of Physical and Mathematical Sciences,  Nanyang Technological University, Singapore  (\email{daoductu001@ntu.edu.sg, hmkiah@ntu.edu.sg, lingsan@ntu.edu.sg}).}
\and Han Mao Kiah\footnotemark[3]
\and San Ling\footnotemark[3]
\and Hengjia Wei\thanks{Department of Electrical and Computer Engineering, Ben-Gurion University of the Negev, Israel
  (\email{hjwei05@gmail.com}).  This work was done while the author was with School of Physical and Mathematical Sciences,  Nanyang Technological University, Singapore. }
}

\usepackage{amsopn}

\begin{document}

\maketitle

% REQUIRED
\begin{abstract}
   A robust positioning pattern is a large array that allows a mobile device to locate its position
by reading a possibly corrupted small window around it. In this paper, we provide constructions 
of binary positioning patterns, equipped with efficient locating algorithms, that are robust to a constant number of errors and have redundancy
within a constant factor of optimality. Furthermore, we modify our constructions to correct rank errors and obtain binary  positioning patterns robust to any errors of rank less than a constant number. Additionally,  we  construct $q$-ary robust positioning sequences robust to a  large  number of errors, some of which have length attaining the upper bound.

Our construction of binary positioning sequences  that are robust to a constant number of errors has the least known redundancy amongst those explicit constructions with efficient locating algorithms. On the other hand, for binary robust positioning arrays,
our construction is the first explicit construction whose redundancy is within a constant factor of optimality.
The locating algorithms accompanying both  constructions run in time cubic in sequence length or array dimension. 
\end{abstract}

% REQUIRED

% REQUIRED

\section{Introduction}
Consider the problem of determining the {\em global} position of a mobile device in a wide environment by simply sensing a small {\em local} area around the device. 
This problem is fundamental in robotics and have practical applications in areas, such as robot localization \cite{Sheinerman:2001}, camera localization   \cite{Szentandrasi:2012}, 3D surface imaging by structured light \cite{Geng:2011}, projected touchscreens \cite{Daietal:2014} and smart styli \cite{Stylus}. 

A classic solution is via the use of positioning patterns.  A {\it positioning pattern} is a large array of dimension  $N_1 \times N_2$, in which 
all contiguous  subarrays of dimension $n_1\times n_2$ are distinct from each other. 
The dimension $n_1\times n_2$ is called the {\it strength} of the positioning pattern. In the special case where $N_1=n_1=1$, we refer to the one-dimensional positioning pattern as {\it positioning sequence}. In practical applications, the positioning pattern is embedded in the wide area, and 
the mobile device reads a small {\it window} of the pattern,  i.e., a subword of length $n$ or a subarray of dimension  $n_1\times n_2$.  Then due to the uniqueness of the window's subpattern, we are able to infer the position of the device.
Positioning patterns have been extensively studied \cite{macwilliams1976pseudo, KumarWei:1992, paterson1994perfect, Mitchelletal:1996, etzion1988constructions}  and 
classical examples include de Bruijn sequences, $m$-sequences, perfect maps (also known as de Bruijn tori) and pseudorandom arrays.

In reality, physical devices are prone to error and we want to locate a device  even when we read a small window erroneously. To this end, we study a class of positioning patterns, called {\it robust positioning patterns}, 
where the subpatterns in distinct windows are far apart from each other. 
In other words, the subpatterns in all windows of a robust positioning pattern form an error-correcting code. 

The study on robust positioning focuses on arrays or sequences in the Hamming metric and history can be traced back to the work of Kumar and Wei \cite{KumarWei:1992} on the minimum distance of partial periods of an $m$-sequence. Recently, Berkowitz and Kopparty \cite{BerkowitzKopparty:2016} presented explicit constructions of robust positioning patterns, along with  {\em efficient} locating algorithms. 
In particular, Berkowitz and Kopparty constructed high-rate $q$-ary robust positioning patterns (both one- and two-dimensional patterns) that locate a position even if a constant {\em fraction} of entries in a window are erroneous.
In the regime where a constant {\em number} of errors are present in a window, the authors provided constructions with redundancy within a constant factor of optimality when the alphabet size is sufficiently large.
When $q=2$, the authors provided one-dimensional positioning patterns robust to a constant number of errors, 
but the result relies on the existence of suitable Mersenne-like primes. 
Efficient positioning in binary two-dimensional patterns robust to a constant number of errors remains open.

In this paper, we study both  binary positioning patterns and $q$-ary positioning sequences.  For {\em binary} positioning patterns that are robust to a constant {\em number} of errors,
without relying on any unproven conjectures, we provide constructions for both one- and two-dimensional patterns 
whose redundancies are within a constant factor of optimality 
(in fact, we reduce the constant factor in the case for the one-dimensional pattern).
Along with these patterns, we propose efficient locating algorithms with complexity $O(n^3)$ or $O((n_1n_2)^3)$, 
where $n$ or $n_1\times n_2$ is the strength of the pattern.
Our construction is based on  $d$-auto-cyclic vectors, Reed-Solomon codes and Gray codes, and can be further modified to correct  errors of {\em rank} less than a constant {\em number}.  For {\it $q$-ary}  positioning sequences,  we modify  Berkowitz and Kopparty's construction to produce sequences robust to larger fraction of errors. We also determine the maximum length of some robust positioning sequences when the distance is large enough.

\section{Preliminaries and Contributions}

For integers $i,j$ with $i<j$, let $[i,j]$ denote  the set of integers $\{i,i+1,i+2,\ldots,j\}$. 
For an integer $N\geq 2$,  let $\bbracket{N}$ denote the set $[0,N-1]$. 
Let $\Sigma$ be an alphabet with $q$ symbols and we index an array of dimension $N_1\times N_2$ using the set $\bbracket{N_1}\times\bbracket{N_2}$.
In particular, for an array $\bA= (a_{ij}) \in \Sigma^{N_1\times N_2}$, we use $\bA[i,i+n_1-1][j,j+n_2-1]$ to denote the $n_1\times n_2$ cyclical contiguous subarray of $\bA$ whose top-left cell is $a_{ij}$; in the one-dimensional  case, for a sequence $\bs=s_0s_1s_2\cdots s_{N-1}  \in \Sigma^N$, we use $\bs[i,i+n-1]$ to denote the length-$n$ cyclical contiguous subword of $\bs$ starting at $s_{i}$.

Denote the Hamming weight of a matrix $\bV$ by $w_H(\bV)$.
For two matrices $\bV$ and $\bW$ of the same dimension $N_1\times N_2$, let $\agree(\bV,\bW)$ be the number of positions at which the corresponding entries  are the same and $d_H(\bV,\bW)$ be the Hamming distance between them.
In other words, $\agree(\bV,\bW)+d_H(\bV,\bW)=N_1N_2$. 

Without loss of generality, we assume that $n_1 \leq  n_2$.  For an $n_1\times n_2$  window, define its {\it area}  to be $n_1n_2$ and its {\it thickness} to be  $(\log_q n_1)/(\log_q n_2)$.

A $q$-ary {\it robust positioning array} (RPA) of strength $n_1\times n_2$ and distance $d$  
is an array $\bA$ over $\Sigma$ in which every pair of rectangular subarrays of dimension  $n_1\times n_2$ is of Hamming distance at least $d$ apart. 
In other words, $d_H(\bA[i,i+n_1-1][j,j+n_2-1],\bA[i',i'+n_1-1][j',j'+n_2-1])\ge d$ for all distinct $(i,j),(i',j')\in \bbracket{N_1-n_1+1}\times  \bbracket{N_2-n_2+1}$.
We denote such array as an $(n_1\times n_2, d)_q$-RPA. %where $q$ can be suppressed when $q=2$. 
For an $(n_1\times n_2,d)_q$-RPA of dimension  $N_1\times N_2$, 
define its {\it rate} to be $(\log_q N_1N_2)/(n_1n_2)$
 %$\frac{\log_q (N_1N_2)}{n_1n_2}$ 
 and define its {\it redundancy} to be $n_1n_2-\log_q (N_1N_2)$.
 Given $q$, $n_1$, $n_2$ and $d$, we are interested in the minimum redundancy of 
 an $(n_1\times n_2,d)_q$-RPA of dimension  {$N_1\times N_2$} and 
 denote this quantity by $\red_q(n_1\times n_2,d)$. When $q=2$, we suppress $q$ in the notation.
  
Since all the subarrays of dimension  $n_1\times n_2$ in an $(n_1\times n_2,d)_q$-RPA  form an error-correcting code of size $(N_1-n_1+1)(N_2-n_2+1)$ with minimum distance $d$,  we have the following bound on $\red_q(n_1\times n_2,d)$.

\begin{proposition}[Sphere-packing Bound]\label{sphere-packing-bound}
For all $q,n_1,n_2$ and $d$, we have that  $$\red_q(n_1\times n_2,d)\ge t\log_q (n_1n_2) + O(1),$$
where $t=\floor{(d-1)/2}$.
\end{proposition} 
 
In the special where $N_1=n_1=1$, we refer to the one-dimensional  $q$-ary  robust positioning array of strength $1\times n$ and distance $d$  as {\it robust positioning sequence }(RPS) and denote it as $(n,d)_q$-RPS. The maximum length of an $(n,d)_q$-RPS is denoted by $P_q(n,d)$. So the minimum redundancy $\red_q(n,d)=  n-\log_q P_q(n,d)\geq t\log_q n+O(1)$, where $t=\floor{(d-1)/2}$.  
%{\color{blue} Berkowitz and Kopparty \cite{BerkowitzKopparty:2016} pointed out that a simple use of Lov\'asz Local Lemma demonstrates 
%that $\red_q(n,d)\leq C_d \log_q n$ for some constant dependent on $d$.}
%

\subsection{Previous Work}\label{sec:prev}
\subsubsection*{One-dimensional RPS} De Bruijn sequences and $m$-sequences are examples of positioning sequences. 
Decodable de Bruijn sequences can be found in Mitchell {\em et al.}~\cite{Mitchelletal:1996}. 
In 1992, Kumar and Wei \cite{KumarWei:1992} studied $m$-sequences with error-correcting ability.  
Using random irreducible linear feedback shift register sequences, they showed the existence of a binary sequence of length $2^n-1$ in which any pair of subwords of length approximately $n+d\log n$ has Hamming distance at least $d$ (and at most $2d$) for $d\leq \sqrt{n}$. Notably, this shows the existence of a sequence that achieves the GV bound whenever $d\leq \sqrt{n}$. 
In 2008, Hagita et al.\,\cite{Hagitaetal:2008} presented constructions for almost optimal $(n,3)$-RPSs. 
However their constructions are based on  a  conjecture on the existence of a certain type of primitive polynomials. 
In these constructions, no efficient locating algorithm was provided.

Recently, Berkowitz and Kopparty \cite{BerkowitzKopparty:2016} presented explicit constructions of robust positioning sequences with efficient locating algorithms. 
For $q=n+1$, they constructed positioning  sequences of length $q^{n-3d-O(1)}$.
%In other words, $\red_q(q+1,d)\ge $
%while both the sphere-packing bound and the Singleton bound show that $P_q(n,d)\leq q^{n-d+O(1)}$.
For binary robust positioning sequences, they  proposed an ``augmented" code concatenation scheme and constructed a class of  binary sequences with constant relative distance $\delta$.  
They also studied  binary sequences with constant distance $d$. 
However, their result relies on an open conjecture on the existence of suitable Mersenne-like primes.
Furthermore, assuming the correctness of the conjecture, the redundancy of the $(n,d)$-RPS in their construction
is at least $9d\log n$.
More recently, Wang {\em et al.} studied the problem under a probabilistic noise model and provided efficient
algorithms to locate the position with high probability.

%More recently, Wang {\em et al.}  \cite{Wangetal:2017} proved that the Gilbert-Varshamov (GV) bound can be attained. Specifically,  using the probabilistic method, they showed that $\red(n,d)\le d\log n+\log n +o(1)$. 
%%$$P(n,d) \geq \frac{2^n}{16n\sum_{i=0}^{d} {{n}\choose {i}}}+n-1.$$  
%In addition, Wang {\em et al.} studied the problem under a probabilistic noise model and provided efficient
%%{\color{red}efficient  (Yes, the complexity is $O(n\log n)$)}  
%algorithms to locate the position with high probability. 

%\vspace{3mm}

\subsubsection*{Two-dimensional RPA} 
When $d=1$, perfect maps and pseudorandom arrays have been studied extensively  
as the two-dimensional generalization of de Bruijn sequences and $m$-sequence \cite{etzion1988constructions,paterson1994perfect,macwilliams1976pseudo}. 
For large values of $d$, Bruckstein {\em et al.} \cite{bruckstein2012simple} constructed a class of binary RPAs
that correctly finds the location provided
less than a quarter of the bits in each row and less than half of the bits in each column are in error.  
Berkowitz and Kopparty \cite{BerkowitzKopparty:2016} provided efficient constructions of  high rate, constant relative distance RPAs over large $q$-ary alphabets. 
In the same paper, they mentioned that these $q$-ary arrays can be used to 
construct binary RPAs of high rate and constant relative distance. 
They also remarked that their methods  {were} unable to construct RPAs with optimal redundancy 
when the distance is constant.

\subsection{Our Contributions}

We provide explicit constructions for binary  RPSs and RPAs with efficient locating algorithms for fixed $d$.
The locating algorithms run in time cubic in window length or window area,
independent of the distance $d$. We also construct RPSs of high rate and  asymptotically optimal RPSs  when $d$ is large. 
Our contributions are as follow.
\begin{enumerate}
\item[(A)] %For $q=2$ and $d=O(1)$, 
In Section~\ref{Sect:biRPSconstantd}, we provide an explicit construction of $(n,d)$-RPSs  with redundancy at most $3 d\log n +6.5 \log n+O(1)$, along with an efficient locating algorithm of complexity $O(n^3)$,
for  fixed $d$.
This improves on Berkowitz and Kopparty's \cite{BerkowitzKopparty:2016} construction 
that requires $9d\log n$ redundancy. 
Note that the  sphere-packing bound suggests that the redundancy is 
$\lfloor \frac{d-1}{2}\rfloor \log n-O(1)$.
%Our explicit construction and efficient locating algorithm only  sacrifice a bit of redundancy.  }
\item[(B)] %For $q=2$, $d=O(1)$ 
Let $\bW$ be a  window of area $A$ and  thickness bounded by a constant.
In Section~\ref{Sect:biRPAconstantd}, we provide  an explicit construction of   binary RPAs for $\bW$ with redundancy at most $4.21d\log A +{36.89}\log A+o(1)$, along with  an efficient locating algorithm of complexity $O(A^3)$, for fixed distance $d$.  This is the first infinite family of RPAs with efficient locating algorithms whose redundancy is within a constant factor of the optimality, i.e.,  $\lfloor \frac{d-1}{2}\rfloor \log A -O(1)$.  In Section~\ref{Sect:biRPAconstantrd}, this construction is modified to produce positioning arrays which are robust to any errors of rank no more than a constant number. 
\item[(C)]
In Section~\ref{Sect:qaryRPS}, we modify the construction of Berkowitz and Kopparty for  $(n,\delta n)_q$-RPSs by doubling the size of the alphabet. The relative distance $\delta$  is improved from $\max\left\{\frac{1-R}{3}, 1-3R\right\}$ to $\max\left\{\frac{1-R}{2}, 1-2R\right\}$, where $R$ is the rate. In contrast, the upper bound on the relative distance is $1-R+o(1)$.
\item[(D)]
We determine the exact value of $P(n,d)$  for $d \geq \floor*{2n/3}$ in Section~\ref{Sect:optimalbinaryRPS} and construct a class of asymptotically optimal $(n,n-1)_q$-RPS for $q=\Omega(n^{2+\epsilon})$ in Section~\ref{Sect:optimalqaryRPS}.
\end{enumerate}

\subsection{Our Approach for Fixed $d$}
We describe the high-level ideas behind our construction of  $(n,d)$-RPSs for fixed $d$. 
Following Berkowitz and Kopparty \cite{BerkowitzKopparty:2016}, 
we pick a $q$-ary code $\mathcal C$ whose block length corresponds to the window length
and we concatenate the codewords of ${\mathcal C}$ in some ordering to obtain our RPS. 
Hence, whenever the window coincides with a possibly erroneous codeword, 
we simply leverage on the error-correcting capability of ${\mathcal C}$ to locate the window.
The main challenge comes when the window does {\em not} coincide with a codeword.
To overcome this, we borrow the following tools.
\begin{enumerate}[(i)]
\item {\em Gray codes}. We use Gray codes to order the codewords of ${\mathcal C}$ so that certain windows of the sequence are of high Hamming distance apart. In fact, this method was used by Berkowitz and Kopparty  to construct $q$-ary RPSs of high rates.
\item {\em Markers}. To construct binary RPS, Berkowitz and Kopparty mapped the $q$-ary symbols of ${\mathcal C}$ to binary strings. 
Then they inserted short binary strings called markers into the binary sequence. 
These markers then allows one to locate the window's position relative to the codewords in ${\mathcal C}$. To further reduce redundancy, our construction utilizes a {\em $d$-auto-cyclic vector} as the marker. 

We remark that $d$-auto-cyclic vectors were introduced by Levy and Yaakobi \cite{LeviYaakobi:2017} in the context of DNA-based data storage. 
In the latter application, one objective is to design a set of primer sequences 
whose prefixes and suffixes satisfy certain distance property (see Yazdi et al. \cite{Yazdietal:2018} for more details).
Not surprisingly, $d$-auto-cyclic vectors, which are useful in the primer sequence design, turn out to be a crucial ingredient of our construction.
\end{enumerate}

\section{Binary Robust Positioning Sequences with Constant Distance}\label{Sect:biRPSconstantd}
In this section, for fixed values of $d$, we propose an explicit construction for an $(n,d)$-RPS 
whose redundancy is $3 d\log n +6.5 \log n+O(1)$.
As our construction is rather intricate, we first present the general ingredients required for constructing an RPS, and 
later provide the specific parameters to achieve the desired redundancy.

First, we review Berkowitz and Kopparty's construction \cite{BerkowitzKopparty:2016}.
Let $\bv^i$ denote the concatenation of $i$ copies of the vector $\bv$, 
and $\bv\bw$ denote the concatenation of two vectors $\bv$ and $\bw$. 
Let $\cC$ be an error-correcting code of length $n$ and minimum distance $d$.
Berkowitz and Kopparty picked {\em certain} words $\bs_0,\bs_1,\ldots,\bs_{M-1}$ from $\cC$
and concatenated them in some order to form a long sequence $\bcs=\bs_0\bs_1\cdots\bs_{M-1}$ of length $N=Mn$.
Notice that from the choice of $\cC$, we have that $d_H(\bs_i,\bs_j)\ge d$ for $0\le i<j\le M-1$.
In fact, via a careful choice of subwords and ordering, Berkowitz and Kopparty \cite{BerkowitzKopparty:2016}
are able to guarantee a certain distance property for {\em all} pairs of subwords in $\bcs$.
We modify this technique to obtain a sequence with weaker property, 
where we guarantee the distance property for {\em some} pairs of subwords in $\bcs$.
Formally, we have the following definitions.

%Now, given a window or subword $\bcs[i,i+n-1]$ for $0\le i\le N-n$, our locating task is to determine $i$.
%To do so, we write $i=an+b$, where $b=i\ppmod{n}$, and we determine $a$ and $b$ separately.
%We have the following formal definitions.

\begin{definition}
Let $\bcs$ be a sequence.  
For two subwords $\bcs[i,i+n-1]$ and $\bcs[j, j+n-1]$ of length $n$ in $\bcs$, 
we say that they \textit{start at the same modular position} if $i\equiv j \ppmod{n}$; 
otherwise, they \textit{start in different modular positions}.

A sequence $\bcs$ is called a \textit{$q$-ary modular robust positioning sequence of  strength $n$ and distance $d$}, 
or \textit{$(n,d)_q$-MRPS} for short, if %if every pair of distinct subwords of length $n$ in the same modular position is of Hamming distance at least $d$ apart.
\[ d_H(\bu,\bv)\ge d \mbox{ for $\bu,\bv$  in the same modular position.}\]
\end{definition}

Next, to construct binary RPS, Berkowitz and Kopparty \cite{BerkowitzKopparty:2016} used a short binary string called {\em marker} and a special mapping to transform symbols from a large alphabet to binary strings.
Their construction then required at least $9d \log n$ bits of redundancy and is reliant on an open conjecture about Mersenne primes.
To reduce the redundancy, we utilise	another marker sequence and introduce the notion of $d$-auto-cyclic vectors. 
%the technique in the work of {Levi and Yaakobi \cite{LeviYaakobi:2017}}

\begin{definition}[Levy and Yaakobi \cite{LeviYaakobi:2017}]
A vector $\bu\in\Sigma^\ell$ is a \textit{$d$-auto-cyclic vector} if $$d_H(\bu,0^i\bu[0,{\ell-i-1}])\geq d$$ for all $1\leq i \leq d$.
\end{definition}

Levy and Yaakobi  provided the following construction of $d$-auto-cyclic vectors.

\begin{proposition}[Levy and Yaakobi  \cite{LeviYaakobi:2017}]\label{eg-autocyclic}
Let  $\ell=d \lceil \log d \rceil+2d$. Set $\bu$ to be the vector
\begin{equation}
\label{dautoc}
\begin{split}
& \bu  =1^d\bu_0\cdots \bu_{\lceil \log d \rceil}, \mbox{ where}\\
 & \bu_i  =((1^{2^i}0^{2^i})^d)[0,d-1].
\end{split}
\end{equation}
Then $\bu$ is a $d$-auto-cyclic vector.
\end{proposition}

\begin{example}
For $d=3$, the sequence $\bu=111 \ 101 \ 110\ 111$ is a $3$-auto-cyclic vector.  
\end{example}

We also introduce the notion of window weight limited.

\begin{definition}[Levy and Yaakobi \cite{LeviYaakobi:2017}]
Let $N,k,d$ be positive integers such that $d<k<N$. We say a vector $\bv\in \F_2^N$ satisfies the $(d,k)$-\textit{window weight limited (WWL)} constraint, and is called a \textit{$(d,k)$-WWL vector}, 
if  $w_H(\bv[i,i+k-1])\geq d$ for any $0\le i\le N-k$.
\end{definition}

We are ready to present our construction.
\vspace{2mm}

\noindent{\bf Construction 1}.
%\subsection{General Construction}
%In this section we propose an explicit  construction for binary positioning sequences with constant distance. 
Given $n$ and $d$, choose $k$ such that $\ell<k$ and $k+\ell<n$, where $\ell=d \lceil \log d \rceil+2d$.
Let $\bu$ be a $d$-auto-cyclic vector of length $\ell$ (e.g., the vector from Proposition~\ref{eg-autocyclic}) and set {$\bp=0^{k}\bu$} to be a vector of length 
{$\ell_p=k+\ell$}. In addition, set $n'=n-\ell_p$.  
Our construction comprises the sequence $\bp$ and a list of length-$n'$  binary vectors $\bs_0,\bs_1,\ldots,\bs_{M-1}$ satisfying the following conditions:
\begin{enumerate}
\item[(P1)] $\bs_i$ is a $(d,k)$-WWL vector for $i\in\bbracket{M}$; 
\item[(P2)] $\bs_{i+1}[0,j-1]\bs_i[j,n'-1]$ is a $(d,k)$-WWL vector for $i\in\bbracket{M}$ and $j\in\bbracket{n'-1}$; and 
\item[(P3)] the concatenation $\bs_0\bs_1\bs_2\cdots\bs_{M-1}$ is an $(n',d)$-MRPS.
\end{enumerate}
Set $\bcs\triangleq \bp\bs_0\bp\bs_1\bp\bs_2\cdots\bp\bs_{M-1}$.

\vspace{2mm}

In the next subsection, we specify the values of $k$ and $\ell$ and provide an explicit method to construct $\bs_i$'s. Consequently, we obtain the sequence $\bcs$ and show that it has the desired redundancy. 
Prior to this, we prove that $\bcs$ is indeed an $(n,d)$-RPS.
Note that (P3) implies $\bcs$ is an $(n,d)$-MRPS. Hence, it remains to show that every two subwords in different modular positions have distance at least $d$. To do so, we have the following technical lemma.

\begin{lemma}\label{pcompare}
Consider the subword $\bw=\bcs[i_0,i_0+n-1]$ in $\bcs$. Pick $i\in \bbracket{n}$. Then the following hold.
\begin{enumerate}[(i)] 
\item If $i+i_0\equiv 0 \ppmod{n}$, then $\bw[i,i+\ell_p-1]=\bp$.
\item If $i+i_0\not \equiv 0 \ppmod{n}$, then $d_H(\bw[i,i+\ell_p-1],\bp)\geq d$.
\end{enumerate}
\end{lemma}

\begin{proof}
Let $\hai$ be the unique integer of $\bbracket{n}$ such that $\hai+i_0\equiv 0\ppmod{n}$. We consider the vector $\bv$, which is obtained by shifting $\bw$ cyclically leftwards $\hai$ times.  Then it suffices to show that $\bv[0,\ell_p-1]=\bp$ and $d_H(\bv[i,i+\ell_p-1],\bp)\geq d$  for $i\in[1,n-1]$. 
Suppose  that $i_0=an+\bai$ where $\bai \in [n]$. From Construction~1 (see  Fig.~\ref{fig-SequenceShift}), 
we have that
\begin{equation*}
\bv = 
\begin{cases} 
\bp\bs_{a}, &\text{if $\bai \leq \ell_p$;} \\
\bp\bs_{a+1}[0,\bai-\ell_p-1]\bs_{a}[\bai-\ell_p,n'-1], & \textup{if $\bai > \ell_p$.}\\
\end{cases}
\end{equation*}
Hence $\bv[0,\ell_p-1]=\bp$.

\begin{figure*}[htbp]
  \centering
  \includegraphics[width=10cm, trim={0 0 0 0}, clip]{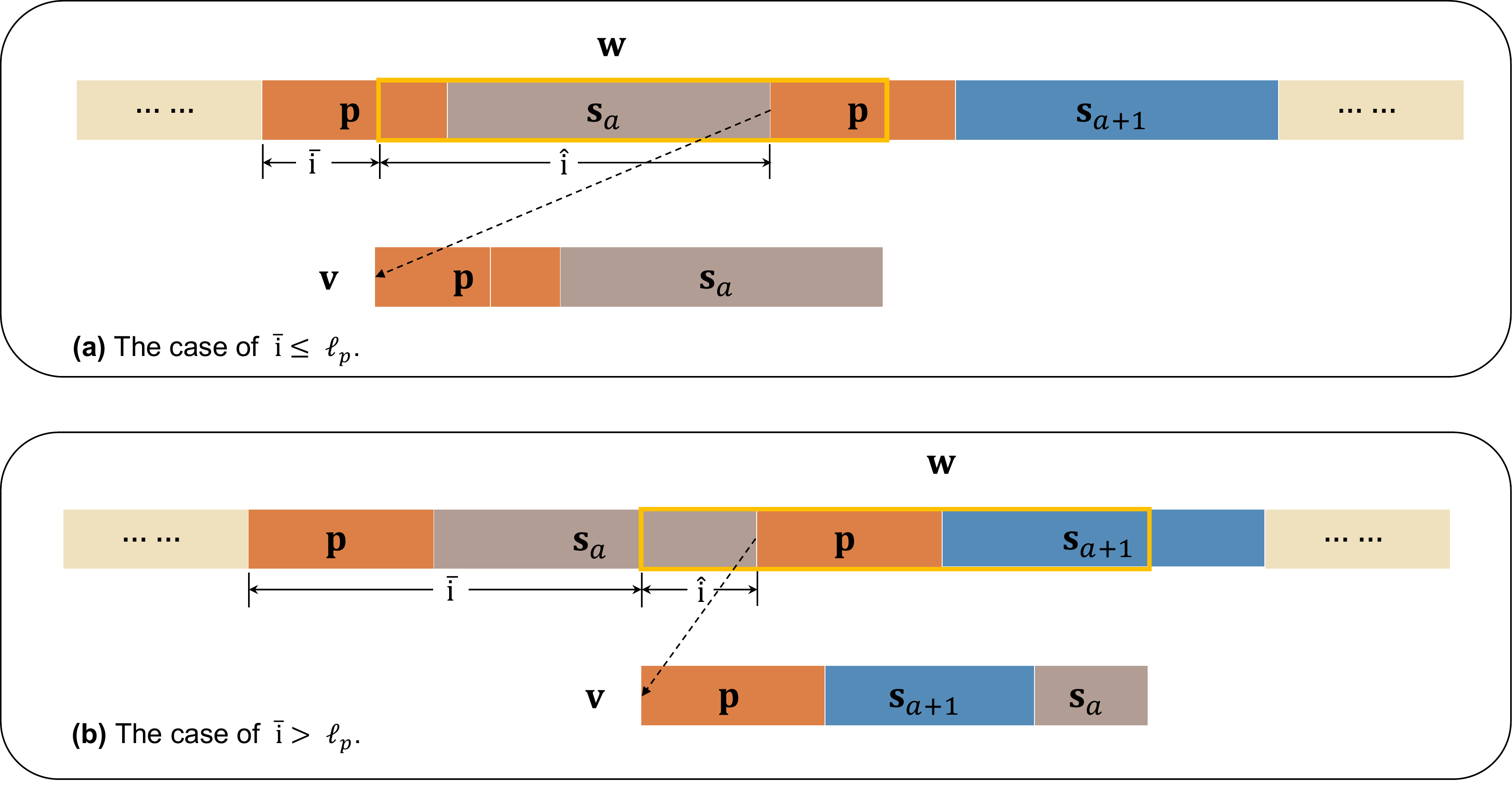}
  \vspace{-4mm}
  
  \caption{The vector $\bv$ obtained by shifting $\bw$ cyclically leftwards $\hai$ times.}
  \vspace{-3mm}
  
  \label{fig-SequenceShift}
\end{figure*}

Now we consider $\bv[i,i+\ell_p-1]$ with $i\not=0$. Since $\bs_a$ and $\bs_{a+1}$ satisfy the conditions (P1) and (P2), we can always assume that $\bv=\bp\bx$ for some $(k,d)$-WWL vector $\bx$ of length $n'$. We proceed by cases.

\vspace{2mm}
\noindent{\bf Case 1}: $i\in [1,d]$. Then
\begin{align*}
\bv[i+k-i,i+k-i+\ell-1]&=\bv[k,k+\ell-1]=\bu,\\
\bp[k-i,k-i+\ell-1]&=0^i\bu[0,\ell-i-1].
\end{align*}
Since $\bu$ is $d$-auto-cyclic, we have
\begin{align*}
  d_H(\bv[i,i+\ell_p-1],\bp)  & \geq  d_H(\bv[k,k+\ell-1], \bp[k-i,k-i+\ell-1])\\
   & = d_H(\bu,0^i\bu[0,\ell-i-1]) \geq d.
\end{align*}

\vspace{2mm}
\noindent{\bf Case 2}: $i\in[d+1,\ell_p-d]$.  Notice that $\bv=0^k\bu\bx$. Since $\ell <k$, the subword $\bv[i,i+k-1]$ should  contain either the length-$d$ prefix of $\bu$ or the length-$d$ suffix of $\bu$, both of which  are $1^d$. So the weight of $\bv[i,i+k-1]$ is at least $d$. It follows that
\begin{align*}
d_H(\bv[i,i+\ell&_p-1],\bp)  \geq d_H(\bv[i,i+k-1],\bp[0,k-1]) =d_H(\bv[i,i+k-1],0^k)\geq d.
\end{align*}

\vspace{2mm}
\noindent{\bf Case 3}: $i\in[\ell_p-d+1,n-k]$.  The subword $\bv[i,i+k-1]$ is contained in $1^d\bx$. Since $\bx$ is a $(d,k)$-WWL vector, the weight of $\bv[i,i+k-1]$ is at least $d$. Again, we have
\begin{align*}
d_H(\bv[i,i+\ell&_p-1],\bp)  \geq d_H(\bv[i,i+k-1],\bp[0,k-1])=d_H(\bv[i,i+k-1],0^k)\geq d.
\end{align*}

\vspace{2mm}
\noindent{\bf Case 4}: $i\in[n-k+1, n-d]$. Since  $i+k-n\geq 1$ and $i+k+d-1-n\leq k-1$, we have
 $\bv[i+k,i+k+d-1]=\bv[i+k-n,i+k+d-1-n]=0^{d}$. Note that $\bp[k,k+d-1]=\bu[0,d-1]=1^d$. It follows that $d_H(\bv[i,i+\ell_p-1],\bp)\geq d$.

\vspace{2mm}
\noindent {\bf Case 5}: $i \in [n-d+1,n-1]$.  Let $\delta=n-i$, then $\delta\in[1,d-1]$. We have
\begin{align*}
\bv[i+k,&i+k+\ell-1]  = \bv[n+k-\delta, n+k+\ell-\delta-1] \\ 
& =  \bv[k-\delta, k+\ell-\delta-1] =  0^\delta\bu[0,\ell-\delta-1].
\end{align*}
Since $\bp[k,k+\ell-1]=\bu$ and $\delta\in[1,d-1]$,  we have 
\begin{align*}
 d_H(&\bv[i,i+\ell_p-1],\bp) \geq d_H(\bv[i+k,i+k+\ell-1],\bp[k,k+\ell-1]) \geq d,
\end{align*}
which completes the proof. 
\end{proof}

Next, we prove that the construction is correct.

\begin{theorem}
Let $\bcs$ be the sequence constructed in Construction 1. Then $\bcs$ is an $(n,d)$-RPS.
\end{theorem}

\begin{proof}  Let $\bw_1$ and $\bw_2$ be two distinct subwords of length $n$ in $\bcs$.
Assume that $\bw_1=\bcs[i,i+n-1]$ and $\bw_2=\bcs[j,j+n-1]$, where $i\not = j$.  
Since $\bs_0\bs_1\cdots \bs_{M-1}$ is an $(n',d)$-MRPS, we have that $\bcs$ is an $(n,d)$-MRPS. 
Hence, $d_H(\bw_1,\bw_2)\geq d$  whenever $i\equiv j \pmod n$. 

It remains to consider the case where $i\not\equiv j \pmod n$. 
Let $\hai$ be the integer of $[n]$ such that $i+\hai\equiv 0 \ppmod{n}$. 
Hence, we have $j+\hai\not \equiv 0 \ppmod{n}$. 
Lemma~\ref{pcompare}  implies that   $\bw_1[\hai,\hai+\ell_p-1]=\bp$  and $d_H(\bw_2[\hai,\hai+\ell_p-1],\bp)\geq d$. 
Hence, $d_H(\bw_1,\bw_2)\geq d_H(\bw_1[\hai,\hai+\ell_p-1], \bw_2[\hai,\hai+\ell_p-1])\geq d$.
\end{proof}

\subsection{Sequence Construction}
Given $n$ and $d$, we provide the choice of $k$ and $\ell$ and 
construct the vectors $\bu, \bs_0, \bs_1,\ldots, \bs_{M-1}$ to satisfy conditions (P1), (P2), and (P3).

First, we set $\ell$ and $\bu$ as in \eqref{dautoc}.
Next, set {$m\triangleq ({3}/{2})\log n$, $k=3m$} and $q\triangleq2^m$.  
Let $\alpha_0,\alpha_1,\ldots,\alpha_{q-1}$ be the $q$ distinct elements in $\F_q$
and let $\phi$ be an arbitrary bijection from $\F_q$ to $\F_2^m$. 
Set $X\triangleq \{\alpha_j\in \F_q: w_H(\phi(\alpha_j))\geq d\}.$
Let $r=\sum_{i=d}^m \binom{m}{i}$, which is  the cardinality of $X$.
Recall that {$n'=n-(k+\ell)$} and our objective is to construct an $(n',d)$-MRPS.
To this end, we require the concepts of Gray codes and Reed-Solomon codes.

\begin{definition}[Gray codes]
Let $\Sigma$ be an alphabet with $q$ symbols and  $\cG =(\bsg_0,\bsg_1,\ldots,\bsg_{q^n-1})$ be a sequence of all the vectors in $\Sigma^n$.  Then $\cG$ is called an {\em $(n,q)$-Gray code} if any two adjacent vectors $\bsg_i$ and $\bsg_{i+1}$ in $\cG$ differ in only one position.
\end{definition}

\begin{theorem}[Decoding for Gray Codes~\cite{Guan:1998}]
Let $q$ and $n$ be two positive integers. There exists an $(n,q)$-Gray code $\cG =(\bsg_0,\bsg_1,\ldots,\bsg_{q^n-1})$ and 
a decoding function $\decG:\F_q^n\to \bbracket{q^n}$ such that $\decG(\bsg_i)=i$ for all $i\in\bbracket{q^n}$.
Furthermore, $\decG$ can be computed in $O(n\log^2 q)$ time.
\end{theorem}

\begin{theorem}[Reed-Solomon code~\cite{WelchBerlekamp:1986}]\label{thm:rs}
Let $q$ be a prime power. Suppose that $k_R<n_R\le q$.
Then there exists a linear code $\Crs$ of length $n_R$, dimension $k_R$ and minimum distance $d_R\triangleq n_R-k_R+1$.
Furthermore, there exist encoding function $\enc^{(n_R,k_R)}:\F_q^{k_R}\to \Crs$ and decoding function $\dec^{(n_R,k_R)}:\F_q^{n_R}\to \Crs$ such that the following hold.
\begin{enumerate}[(i)]
\item For all $\bsg\in \F_q^{k_R}$, the $k_R$-prefix of $\enc^{(n_R,k_R)}(\bsg)$ is $\bsg$. In other words, $$\enc^{(n_R,k_R)}(\bsg)[0,k_R-1]=\bsg.$$
\item Choose $\bc\in\Crs$ and suppose that $d_H(\bav,\bc)\le (d_R-1)/2$. Then $\dec^{(n_R,k_R)}(\bav)=\bc$. Furthermore, $\dec^{(n_R,k_R)}$ can be computed in $O(n^3)$ time.
\end{enumerate}
\end{theorem}

In their construction for $q$-ary RPSs, Berkowitz and Kopparty  \cite{BerkowitzKopparty:2016} used a Gray code to give an ordering to a subset of codewords in a Reed-Solomon code, and concatenated these codewords in this ordering to form the desired sequence. In our construction, we adapt the technique to obtain a family of $q$-ary vectors $\bc_i$ such that  $\bc_0\bc_1\cdots \bc_{M-1}$ is a $q$-ary MRPS. 
Then we apply the mapping $\phi$ to each $\bc_i$ and append short sequences $1^d$ in proper positions to obtain the WWL vector $\bs_i$.

Specifically, set $n_R\triangleq (n'-(2d+{2})d)/m$ and $k_R\triangleq n_R-2d-2$. Consider the Reed-Solomon code $\Crs$  from Theorem~\ref{thm:rs}.
%Then we have that $k_R <n_R< q$. It follows that for any vector $\bsg\in X^{k_R}$, there is a unique polynomial $f^{\bsg}(x)\in \F_q[x]$ of degree at most $k_R-1$ such that $f^{\bsg}(\alpha_j)=\bsg[j]$ for $0\leq j <k_R$. 
Now we provide our construction of $\bs_i$ for $i\in \bbracket{r^{k_R}}$, and consequently, the sequence $\bcs$.
\vspace{2mm}

\noindent{\bf Construction 1A}.
Let $M=r^{k_R}$ and $\cG=(\bsg_0,\bsg_1,\ldots, \bsg_{M-1})$ be a  $(k_R,r)$-Gray code over $X$. 
For $i \in \bbracket{M}$, set $\bc_i=\enc^{(n_R,k_R)}(\bsg_i)$.
%For each $i \in \bbracket{M}$, set
%\[\bc_i=(f_i(\alpha_0),f_i(\alpha_1),f_i(\alpha_2),\ldots,f_i(\alpha_{n_R-1})),\]
%where $f_i(x)\in \F_q[x]$ is the unique  interpolating polynomial of degree at most $k_R-1$ such that $f_i(\alpha_j)=\bsg_i[j]$ for all $0\leq j <k_R$.
Then for each  $\bc_i$,
construct a binary vector $\bs_i$ as
\begin{align*}
\bs_i=\phi(\bc_i&[0])\phi(\bc_i[1])\cdots\phi(\bc_i[k_R-1])1^d\phi(\bc_i[k_R])1^d\phi(\bc_i[k_R+1])\cdots1^d\phi(\bc_i[n_R-1]).
\end{align*}
Finally, let {$\bp=0^{k}\bu$}, where $\bu$ is the $d$-auto-cyclic vector in~\eqref{dautoc}. 
Construct the sequence $\bcs$ as
\[\bcs=\bp\bs_0\bp\bs_1 \cdots\bp\bs_{M-1}.\]

\begin{table*} 
\begin{center}
\footnotesize
%\caption{Notation involved in the construction of the binary $(n,d)$-RPS}\label{tab.notation}
  %\begin{tabular}{ |c | l ||c | l | c| c | c |c | c |}
  \begin{tabular}{ |c |p{10cm} |}
  \hline
     {\scriptsize Notation} &   Remark                \\ [3pt]   \hline
    $n$ & the strength of the RPS                              \\[3pt]
    $d$ & the distance of the RPS                             \\[3pt]
    $\ell$ & the length of the $d$-auto-cyclic vector    \\ [3pt]
    $m$ & $m\triangleq \frac{3}{2}\log n$                    \\ [3pt]
    $k$       &  $k\triangleq 3m$                                \\ [3pt]
    $\ell_p$ & $\ell_p\triangleq k+\ell$                                      \\ [3pt]
   $r$ & $r\triangleq\sum_{i=d}^m {m\choose i}$                \\   [3pt]
     $n'$ & $n'\triangleq n-\ell_p$                                     \\   [3pt]
     $q$ & $q\triangleq 2^m$                                                \\    [3pt]
     $n_R$ & $n_R\triangleq\frac{1}{m}(n'-(2d+2)d)$                   \\    [3pt]
      $k_R$ & $k_R\triangleq n_R-2d-2$                   \\   [3pt]
       $M$      &     $M\triangleq r^{k_R}$                 \\    [3pt]
	 $\F_q$       & the finite field with $q$ elements       \\    [3pt]
	 $\alpha_j$          &  the element in $\F_q$          \\    [3pt]
	 $\phi$           & a one-to-one map from $\F_q$ to $\F_2^m$                      \\    [3pt]
	$X$       & the set of $\alpha_j$ such that  $w_H(\phi(\alpha_j))\geq d$     \\    [3pt]
	 $\cG$          &  a $(k_R,r)$-Gray code                                        \\    [3pt]
	  $\bsg_i$  &  the $i$-th vector in $\cG$                                   \\    [3pt]
	  $\bc_i$     &  the codeword in a $q$-ary Reed-Solomon code of  length $n_R$     and  dimension $k_R$  such that  $\bc_i[0,k_R-1]=\bsg_i$  \\    [3pt]
	$\bs_i$        & the concatenation of the binary vectors $\phi[\bc_i[j]]$, $j \in [k_R]$,     as well as the vectors  $1^d\phi[\bc_i[j]]$, $k_R\leq j \leq n_R-1$    \\    [3pt]
	       $\bu$         &   the $d$-auto-cyclic vector of length $\ell$    \\    [3pt]
	     $\bp$        &    $\bp\triangleq 0^{k}\bu$     \\    [3pt]
   \hline
  \end{tabular}
  \caption{Notation Summary for Construction 1A}\label{tab.notation}
 \end{center}
\end{table*}

\begin{figure*}[htbp]
  \centering
  \includegraphics[width=12cm, trim={0 0 0 0}, clip]{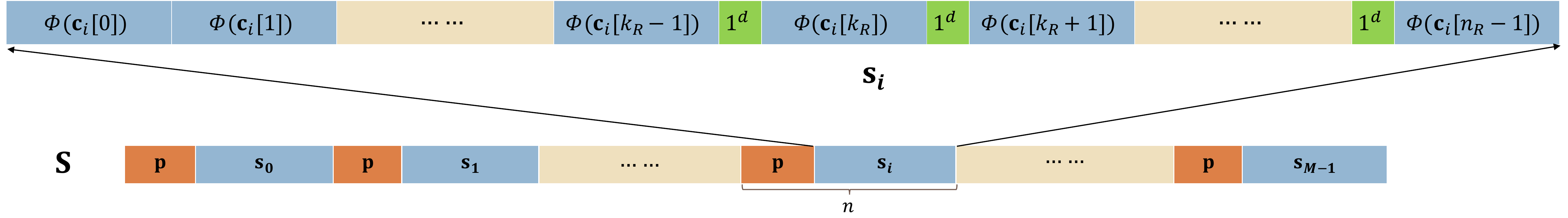}
  \caption{The sequence $\bcs$ in Construction~1A.}
  \label{fig-seqcon}
\end{figure*}

%The following is an  outline of our construction for $\bs_i$:
%\begin{enumerate}
%\item Choose a  subcode $\cC$ (not necessarily cyclic) from a Reed-Solomon code of length $n_R$ over $\F_{2^m}$ and give its codewords an ordering such that their projection on the first $k_R$ positions forms a Gray code with $\sum_{i=d}^m 2^i$ symbols. %We show that if we concatenate the codewords of $\cC$ in this order to form a sequence, then any two subwords of length $n_R$ in the same modular position have distance $\geq d+2$.
%\item Concatenate all the binary vectors of length $m$ using $\cC$ and append a small number of symbols `1' to obtain a binary code $\cS$, such that each $\bs_i$ satisfies the $(d,k)$-WWL constraint.
%\end{enumerate}

We summarise in Table~\ref{tab.notation} all the parameters  and notations involved in our construction (see also Fig.~  \ref{fig-seqcon}). 
To simplify our exposition and  analysis, we assume that all parameters are integers. 

Observe that each  $\bc_i$ is a codeword whose length-$k_R$ prefix is $\bsg_i$. 
Hence, when $j<k_R$, the  symbol $\bc_i[j]$ belongs to $X$ and  $w_H(\phi(\bc_i[j]))\geq d$. 
However, when $j\geq k_R$, the symbol $\bc_i[j]$ may not belong to $X$ and the weight of $\phi(\bc_i[j])$ may be less than $d$. So, we prepend a sequence $1^d$ at the head of  $\phi(\bc_i[j])$ for  each $j\geq k_R$. Since $k=3m\geq m+2d$, it is easy to check that the vectors  $\bs_0, \bs_1, \ldots, \bs_{M-1}$ satisfy the conditions (P1) and (P2). For  (P3), we have the following result on $\bcs$.

\begin{lemma}\label{RSwindow}
The concatenation $\bs_0\bs_1\cdots\bs_{M-1}$ is an $(n',d)$-MRPS.
\end{lemma}

\begin{proof} Since each $\bs_i$ is obtained by inserting the sequences $1^d$ at fixed positions in the concatenation of the binary strings $\phi[\bc_i[j]]$, it suffices to show that 
the concatenation $\bc_0\bc_1\cdots\bc_{M-1}$ is an $(n_R,d+{1})_q$-MRPS.

Assume that $\bw_1$ and $\bw_2$ start at position $i$ and position $j$ respectively. Since $\bw_1$ and $\bw_2$ are in the same modular position, we may assume that $\bam \equiv i \equiv j \pmod{n_R}$, where $\bam\in\bbracket{n_R}$. Further let $i=an_R+\bam$ and $j=bn_R+\bam$.  We proceed by cases.

\vspace{2mm}
\noindent{\bf Case 1}: $\bam=0$. Then $\bw_1=\bc_a$ and $\bw_2=\bc_b$. 
Since $\bc_a$ and $\bc_b$ belong to $\Crs$, %are codewords in a Reed-Solomon code, 
we have that $d_H(\bw_1,\bw_2)=d_H(\bc_a,\bc_b) \geq d_R=n_R-k_R+1>d+1$.

\vspace{2mm}
\noindent{\bf Case 2}: $ \bam \in [1,k_R]$.  So, $\bw_1=\bc_a[\bam,n_R-1]\bc_{a+1}[0,\bam-1]$. Since $\cG$ is a Gray code,  $d_H(\bc_{a}[0,\bam-1], \bc_{a+1}[0,\bam-1])=d_H(\bsg_a,\bsg_{a+1})\leq 1$. In other words, 
$$d_H(\bw_1, \bc_{a}[\bam,n_R-1]\bc_{a}[0,\bam-1])\leq 1.$$
Similarly, we have
$$d_H(\bw_2, \bc_{b}[\bam,n_R-1]\bc_{b}[0,\bam-1] )\leq 1.$$
It follows that  $$d_H(\bw_1,\bw_2)\geq  d_H(\bc_a,\bc_b)-2\geq n_R-k_R-1>d+1.$$

\vspace{2mm}
\noindent{\bf Case 3}: $\bam\in [k_R +1, n_r-1]$. We partition the interval $[0,n_R-1]$ into three pieces by setting
\begin{align*}
I_1&=[0, n_R-\bam-1],\\
I_2&= [n_R-\bam, n_R-\bam+k_R-1], \mbox{ and}\\
I_3&=[n_R-\bam+k_R, n_R-1].
\end{align*}
%\[
%I_1=[0, n_R-\bam-1],~
%I_2= [n_R-\bam, n_R-\bam+k_R-1], \mbox{ and }\\
%I_3=[n_R-\bam+k_R, n_R-1].
%\]
For $j\in\{1,2,3\},$ let $\agree_j=\agree(\bw_1[I_j], \bw_2[I_j])$. Then
\begin{align*}
\agree_2+\agree_3 &= \agree(\bc_{a+1}[0,\bam-1], \bc_{b+1}[0,\bam-1])\\
  & \leq  k_R-1\\
\agree_1+\agree_3 & \leq  |I_1|+|I_3|=n_R-|I_2|=n_R-k_R.
\end{align*}
Since $d_H(\bc_{a}[0,k_R-1]), \bc_{a+1}[0,k_R-1]))\leq 1$ and $d_H(\bc_{b}[0,k_R-1]), \bc_{b+1}[0,k_R-1]))\leq 1$,  we have
\begin{align*}
\agree_1+\agree_2 =&\agree( \bc_{a}[\bam,n_R-1]\bc_{a+1}[0,k_R-1],  \bc_{b}[\bam,n_R-1]\bc_{b+1}[0,k_R-1])\\
\leq &\agree( \bc_{a}[\bam,n_R-1]\bc_{a}[0,k_R-1], \bc_{b}[\bam,n_R-1]\bc_{b}[0,k_R-1])+2\\
\leq &\agree( (\bc_{a},\bc_b) +2\leq  k_R+1.
\end{align*}
Summing these three inequalities yields \[\agree_1+\agree_2+\agree_3\leq \frac{n_R+k_R}{2}.\]
Therefore, we have 
\begin{align*}
d_H(\bw_1,&\bw_2)=n_R-\agree(\bw_1,\bw_2) \geq \frac{n_R-k_R}{2}=d+1,
\end{align*}
which completes the proof. 
 \end{proof}
 
Therefore, Construction 1A yields an $(n,d)$-RPS as desired. 
We analyse the required redundancy in the following corollary.

\begin{corollary}\label{redundancy-binary}
For sufficient large $n$, there is an $(n,d)$-RPS with redundancy  at most $3 d\log n + 6.5\log n+O(1).$
\end{corollary}

\begin{proof}
%{\color{red}Hengjia, please check that the parameters here and in Table 1 are correct. (I've checked and  corrected the parameters in Table 1)}
Recall that
\begin{align*}
r&= \sum_{i=d}^m {m\choose i}=2^m\left[1-\sum_{i=0}^{d-1} {m\choose i}\frac{1}{2^m}\right] \geq 2^m\left[1-\exp\left(-\frac{2(m/2-d+1)^2}{m}\right) \right]\\
& \geq 2^m\left(1-e^{-\frac{2m}{3 \log e}} \right) = 2^m\left(1-\frac{1}{n} \right).
\end{align*}
The first inequality comes from Hoeffding's inequality for the tail of  the binomial distribution while the second one holds as ${(m/2-d+1)^2}/{m}>{m}/{(3\log e)}$ when $m$ is large enough. Then we have
\begin{align*}
\log r & \geq m+\log\left(1-\frac{1}{n} \right)  = m+\ln\left(1-\frac{1}{n}\right)\log e  \geq m-\frac{\log e}{n-1}.
\end{align*}
Note that $k_R=\frac{1}{m}[n-k-\ell-(2d+2)(m+d)]=\frac{n}{m}-2d-5-O\left(\frac{1}{m}\right)$. Hence the redundancy of $\bcs$ is
\begin{align*}
 \ n-  \log(nM)  & =  \ n-k_R\log r-\log n \\
 & \leq  \  n- \left(\frac{n}{m}-2d-5-O\left(\frac{1}{m}\right)\right) \left(m - \frac{\log e}{n-1}\right)-\log n\\
 &= \ 2dm+5m-\log n+O(1) =  \ 3 d\log n +6.5 \log n+O(1). 
\end{align*}
\end{proof}

\noindent{\bf Remark.} In Construction 1A, we convert the $q$-ary vectors into binary vectors by mapping the elements of $\F_q$ to the binary vectors of length $m$, and  we append some short sequences $1^d$ so that the resulting sequences satisfy conditions (P1) and (P2).
Alternatively, one may choose a prime power $q$ that is at most $\sum_{i=d}^m {m\choose i}$ and map the elements in $\F_q$ to the binary length-$m$ vectors with weight at least $d$.
%As another approach, one may also use Theorem ?? to choose a prime $q$
%which is no more than $\sum_{i=d}^m {m\choose i}$ and map the elements in $\F_q$ to the binary length-$m$ vectors with weight at least $d$. 
This approach then results in an $(n,d)$-RPS with redundancy $4.21d \log n+9.53 \log n+O(1)$, which is larger than that in Corollary~\ref{redundancy-binary}. However, in next section, when we construct  2-D robust positioning arrays, we have to adopt this approach as it is difficult to tile some small patterns and large squares while keeping low redundancy.

\subsection{Locating Algorithm}
We present a locating algorithm for the subwords of the sequence $\bcs$ in Construction 1A.
In particular, the locating algorithm corrects up to $\floor{(d-1)/2}$ errors in $O(n^3)$ time, 
independent of parameter $d$.

Suppose that $\bw$ is a subword of $\bcs$ that is corrupted at $e$ positions with $e\le (d-1)/2$.
In other words, there is a unique index $i$ such that $d_H(\bcs[i,i+n-1],\bw)\le (d-1)/2$ and 
our task is to recover $i$.
Equivalently, if we write $i$ as $an+\bai$ with $\bai\equiv i\pmod{n}$, then our task is to recover both $a$ and $\bai$.
In what follows, we give a broad overview of the steps and the detailed implementation of the algorithm is provided in Algorithm~\ref{locatingalgo}.

\begin{enumerate}[(I)]
\item We determine $\bai$. To do so, we determine the unique index $\hai$ such that $d_H(\bw[\hai,\hai+\ell_p-1],\bp)\le (d-1)/2$. Set $\bai\in \bbracket{n}$ such that $\bai+\hai\equiv 0\pmod{n}$
\item Next, we cyclically rotate $\bw$ leftwards by $\hai$ positions to obtain $\bv$. Observe that $\bv$ is the binary image obtained from either a $q$-ary codeword $\bc_a$ or a concatenation $\bc_a[n_R-j+1,n_R]\bc_{a+1}[0,j]$ for some $j\in \bbracket{n_R}$. Since $\bv$ is obtained via the map $\phi$ and prepending the string $\bp$ and inserting $n_R-k_R$ strings $1^d$, we reverse this process to obtain the $q$-ary estimate $\bav$. 
\item Finally, depending on the value of $\bai$, we apply the Reed-Solomon decoding algorithm $\dec$ to find either $\bc_a$ or $\bc_{a+1}$ or some shortened versions of these words. Therefore, we determine $a$ and hence, obtain $i=an+\hai$.
\end{enumerate}

%~\ref{con-binary}.

%\begin{definition}
%Given a vector $\bw$, the operation $L^i(\bw)$ shifts the vector $\bw$ to the left cyclically $i$ positions.
%\end{definition}

\begin{algorithm} [ht]
\begin{algorithmic}
\caption{Locating algorithm for the sequence $\bcs$ in Construction~1A}%~\ref{con-binary}}
\label{locatingalgo}
\STATE {\bf Input}: a sequence $\bw$ of length $n$ %in the $(n,d)$-RPS $\bcs$ %$n$, $d$, $e$, and  
\STATE {\bf Output}: a position $i\triangleq an+\bai$ such that $d_H(\bcs[i,i+n-1],\bw)\le (d-1)/2$
%\FORALL {$i \in [n]$}
%	\IF {$d_H(\bw[i,i+ \ell_p -1],\bp) < e$}
%	\STATE $\hai \leftarrow i$
%	\ENDIF
%\ENDFOR
\STATE{}
\STATE $\hai \leftarrow$ unique index such that $d_H(\bw[\hai,\hai+ \ell_p -1],\bp)\le (d-1)/2$
\STATE Set $\bai \in \bbracket{n}$ such that $\bai+\hai\equiv 0\ppmod{n}$
\STATE{}
\STATE $\bv \leftarrow$ the vector obtained by rotating $\bw$ cyclically leftwards $\hai$ positions
\STATE $\hav \leftarrow $ the vector obtained from $\bv$ by deleting $\bv[0,\ell_p-1]$ and $\bv[\ell_p+mk_R+(m+d)j, \ell_p+mk_R+(m+d)j+d-1]$ for all $j \in [n_R - k_R]$
\STATE $\bav \leftarrow \phi^{-1}(\hav[0,m-1])\phi^{-1}(\hav[m,2m-1])\cdots\phi^{-1}(\hav[m(n_R-1),mn_R-1]) $\\
\STATE{}
%\STATE Let $\bai \in [n]$ such that $\bai+\hai\equiv 0\ppmod{n}$
\IF {$\bai \in [0, \ell_p + mk_R -1]$}
\STATE $\bc\gets\dec^{(n_R,k_R)}(\bav)$
\STATE $a\gets \decG(\bc[0,k_R-1])$
%\STATE Run the decoding algorithm for Reed-Solomon codes of  length $n_R$ and dimension $k_R$ on $\bav$ to get a vector $\bc$
%\STATE Run the encoding algorithm for Gray codes on $\bc[0,k_R-1]$ to get an index $a$
\ELSIF {$\bai \in [\ell_p + mk_R, \ell_p+mk_R+(d+1)(m+d)-1]$}
\STATE $\bav^s \leftarrow$ the shortened codeword $\bav[0,k_R-1]\bav[k_R+(d+1),n_R-1]$
\STATE $\bc^s\gets\dec^{(n_R-(d+1),k_R)}(\bav^s)$
\STATE $a~\gets \decG(\bc^s[0,k_R-1])$
%\STATE Run the decoding algorithm for Reed-Solomon codes of length $n_R-(d+1)$ and  dimension $k_R$ on $\bav^s$ to get a vector $\bc^s$
%\STATE Run the encoding algorithm for Gray codes on $\bc^s[0,k_R-1]$ to get an index $a$ 
\ELSE
\STATE $\bav^s \leftarrow$ the shortened codeword $\bav[0,k_R+(d+1)-1]$
\STATE $\bc^s\gets\dec^{(n_R-(d+1),k_R)}(\bav^s)$
\STATE ${a+1}~\gets \decG(\bc^s[0,k_R-1])$
%\STATE Run the decoding algorithm for Reed-Solomon codes of length $n_R-(d+1)$ and  dimension $k_R$ on $\bav^s$ to get a vector $\bc^s$
%\STATE Run the encoding algorithm for Gray codes on $\bc^s[0,k_R-1]$ to get an index $a+1$ 
\ENDIF
 \RETURN {$an + \bai$} 
\end{algorithmic}
\end{algorithm}

\begin{theorem}\label{thm:1d}
Suppose $\bw$ is a corrupted subword of the sequence $\bcs$ with exactly $e$ errors. 
If $2e < d$, then  Algorithm~\ref{locatingalgo} can determine the position of $\bw$ in the sequence $\bcs$ in $O(n^3)$ time.
\end{theorem}

\begin{proof}
Suppose the corrupted subword $\bw$ starts at position ${\alpha}n+\beta$, where $\alpha\in \bbracket{M}$ and $\beta \in \bbracket{n}$.
Denote the original subword  $\bcs[\alpha n+\beta, (\alpha+1)n+\beta-1]$ as $\bw^\circ$, 
and so, $d_H(\bw,\bw^\circ)=e$. 
Lemma~\ref{pcompare} implies that $d_H(\bw[i,i+\ell_p-1],\bp)   \leq e$ when $i+\beta \equiv 0\ppmod{n}$, and $d_H(\bw[i,i+\ell_p-1],\bp) \geq d-e$ when $i+\beta \not\equiv 0 \ppmod{n}$. 
Since $d-e>e$, the value $\beta$ can be uniquely determined and we have $\bai=\beta$.  
In order to determine $\alpha$, we consider the following  cases.

\vspace{2mm}
\noindent{\bf Case 1}: $\bai \in \bbracket{\ell_p}$. 
By shifting the original subword  $\bw^\circ$ leftwards $\hai$ times, we obtain $\bp\bs_\alpha$. 
Since shifting both $\bw$ and $\bw^\circ$ simultaneously does not increase the Hamming distance, 
we have $d_H(\bv,\bp\bs_\alpha)=e$.
After removing the sequences $\bp$ and  $1^d$ from $\bp\bs_\alpha$ and the corresponding subwords from $\bv$,
we have that $d_H(\hav,\phi(\bc_\alpha[0])\phi(\bc_\alpha[1])\cdots\phi(\bc_\alpha[n_R-1])) \leq e$.
Next, we apply the inverse function $\phi^{-1}$ and observe that the Hamming distance does not increase.
Therefore, $d_H(\bav,\bc_\alpha) \leq e<d/2.$

Since $d_R=2d+2$, we have that $d_H(\bav,\bc_\alpha)\le (d_R-1)/2$, 
applying the decoding algorithm $\dec^{(n_R,k_R)}$ on $\bav$ recovers $\bc_\alpha$.
Finally, we apply $\decG$ to $\bc_\alpha[0,k_R-1]$ to recover $\alpha$ and hence, the output $a$ is indeed $\alpha$.
%For any $b \neq \alpha$, since $k_R = n_R - (2d+2)$, the distance between $\bc_{\alpha}$ and $\bc_b$ is at least $2d + 3$. It follows that
%$$d_H(\bav,\bc_b) \geq d_H(\bc_\alpha,\bc_b) - d_H(\bav,\bc_\alpha) \geq (2d+3)-e.$$
%Therefore, if $e< 2d+3-e$, or $2e < 2d+3$, we can apply the decoding algorithm for Reed-Solomon codes with $\bav$ as input to recover $\bc_\alpha$. Then the encoding algorithm for Gray codes is used to recover $\alpha$ and outputs $a=\alpha$.

\vspace{2mm}
\noindent{\bf Case 2}: $\bai \in [\ell_p, \ell_p+mk_R-1]$. By shifting $\bw^\circ$ leftwards $\hai$ positions, 
we obtain the sequence $\bp\bs_{\alpha+1}[0,\bai-\ell_p-1]\bs_\alpha[\bai-\ell_p,n'-1]$.
So, $d_H(\bv,\bp\bs_{\alpha+1}[0,\bai-\ell_p-1]\bs_\alpha[\bai-\ell_p,n'-1]) = e.$

Set $j = \lfloor {(\bai-\ell_p)}/{m} \rfloor$, and so, $j < k_R$.
Note that there may be a subword of length $m$ which covers both the tail of $\bs_{\alpha+1}[0,\bai-\ell_p-1]$ and the head of $\bs_\alpha[\bai-\ell_p,n'-1]$.
Hence, we have $d_H(\bav,\bc_{\alpha+1}[0,j]\bc_\alpha[j+1,n_R-1]) \leq e + 1.$ 
It follows that
\begin{align*}
d_H(&\bav,\bc_\alpha)  = d_H(\bav,\bc_\alpha[0,j]\bc_\alpha[j+1,n_R-1]) \\ 
& \leq d_H(\bav,\bc_{\alpha+1}[0,j]\bc_\alpha[j+1,n_R-1]) + 1 \leq e+2\le (d_R-1)/2.
\end{align*}
Similar to Case 1, we apply $\dec^{(n_R,k_R)}$ to $\bav$ to recover $\bc_\alpha$ and then apply $\decG$ to recover $a=\alpha$.
%The first inequality holds as $j< k_R$ and $d_H(\bc_{\alpha}[0,k_R-1],\bc_{\alpha+1}[0,k_R-1])\leq 1$.

%For any $b \neq \alpha$, $$d_H(\bav,\bc_b) \geq d_H(\bc_\alpha,\bc_b) - d_H(\bav,\bc_\alpha) \geq 2d-e+1.$$
%Therefore, if $e+2< 2d-e+1$, or $2e<2d-1$, then the same strategy as in Case 1 can be applied to determine $\alpha$.
\vspace{2mm}
\noindent{\bf Case 3}:  $\bai \in [\ell_p + mk_R, \ell_p+mk_R+(d+1)(m+d)-1]$. Similar to Case 2, we have 
$$d_H(\bv,\bp\bs_{\alpha+1}[0,\bai-\ell_p-1]\bs_\alpha[\bai-\ell_p,n'-1]) \leq e.$$
Due to the range of $\bai$, $\bs_{\alpha+1}[0,\bai-\ell_p-1]$ contains the subword $\bs_{\alpha+1}[0,mk_R-1]$ and $\bs_\alpha[\bai-\ell_p,n'-1]$ contains the subword $\bs_\alpha[mk_R+(d+1)(m+d),n'-1]$. It follows that the distance between the shortened vector $\bav^s=\bav[0,k_R-1]\bav[k_R+(d+1),n_R-1]$
 and the vector $\bc_{\alpha+1}[0,k_R-1]\bc_\alpha[k_R+(d+1),n_R-1]$ is  no more than $e.$
%Noting that in this case 
Since $d_H(\bc_{\alpha}[0,k_R-1], \bc_{\alpha+1}[0,k_R-1])\leq 1$, we have
\[d_H(\bav^s,\bc_{\alpha}[0,k_R-1]\bc_\alpha[k_R+(d+1),n_R-1]) \leq e+1.\]
The shortened vector $\bc_{\alpha}[0,k_R-1]\bc_\alpha[k_R+(d+1),n_R-1]$ can be treated as a codeword of a Reed-Solomon code of length $n_R-(d+1)$ and dimension $k_R$. Since $e+1\le (n_R-(d+1)-k_R)/2$, 
we apply the decoding algorithm $\dec^{(n_R-(d+1),k_R)}$ to recover $\bc_\alpha[0,k_R-1]$.
%of Reed-Solomon codes of length $n_R-(d+1)$ to recover $\bc_\alpha[0,k_R-1]$.
As before, we apply $\decG$ to recover $\alpha$.
%the encoding of Gray codes to recover $\alpha$.
% so for any  $b \neq \alpha$, 
%$$d_H(\bc_{\alpha}[0,k_R-1]\bc_\alpha[k_R+(d+1),n_R-1], \bc_{b}[0,k_R-1]\bc_b[k_R+(d+1),n_R-1])\geq n_R-(d+1)-k_R+1=d+2.$$
%Therefore, if $e+1< d+2-(e+1)$, or $2e<d$, we apply the decoding of Reed-Solomon codes of length $n_R-(d+1)$ to recover $\bc_\alpha[0,k_R-1]$ and then use the encoding of Gray codes to recover $\alpha$.

\vspace{2mm}
\noindent{\bf Case 4}: $\bai \in [\ell_p+mk_R+(d+1)(m+d),n-1]$. In this case, we consider the shortened vector 
$\bav^s=\bav[0,k_R+(d+1)-1]$. Similar to Case 3, we have 
\[d_H(\bav^s,\bc_{\alpha+1}[0,k_R+(d+1)-1]) \leq e.\]
%Therefore, if $e< d+2-e$, or $2e<d+2$, we can recover $\bc_{\alpha+1}[0,k_R-1]$ and so $\alpha+1$.
As before, we can $\dec^{(n_R-(d+1),k_R)}$ to recover $\bc_{\alpha+1}[0,k_R-1]$ and hence, recover $\alpha+1$.

\vspace{2mm}

We analyse the running time.
To determine $\hai$, we require $\ell_pn$ comparisons.
Next, the Reed-Solomon decoding  $\dec^{(n_R,k_R)}$ runs in $O(n_R^3)=O(n^3)$ time.
Finally, the decoding of Gray codes $\decG$ runs in $O(k_R\log^2 q)=O(nm^2)=O(n\log^2 n)$ time.
Therefore, Algorithm~\ref{locatingalgo} computes the location in $O(n^3)$ time.
%{\color{red}
%The  complexity of the decoding of Reed-Solomon codes is  $O(n_R^3)$ \cite{WelchBerlekamp:1986}. The complexity of  the encoding of Gray codes is $O(k_R)$ \cite{Guan:1998}. Therefore, in total, the running time of the algorithm is $O(n^3)$.}
%The time complexity of the algorithm mostly comes from the time to run the  {\bf for} loop and the time to decode Reed-Solomon codes and encode Gray codes. The {\bf for} loop requires $\ell_pn$ comparisons. The  complexity of the decoding of Reed-Solomon codes is  $O(n_R^3)$ \cite{WelchBerlekamp:1986}. The complexity of  the encoding of Gray codes is $O(k_R)$ \cite{Guan:1998}. Therefore, in total, the running time of the algorithm is $O(n^3)$.
\end{proof}

\section{Binary Robust Positioning Arrays with Constant Distance}
\label{Sect:biRPAconstantd}
%\subsection{Construction}
%In this subsection, 
Let $\bW$ be an $n_1\times n_2$ window of area $A$ and thickness bounded by a constant. We generalise Construction 1A %Construction~\ref{con-binary} 
to produce binary RPAs for $\bW$ with constant distance $d$.
To this end, we require the following number theoretic result.

\begin{lemma}[Baker et al. \cite{Bakeretal:2001}]\label{prime}
Let $\theta=0.525$. There exists $x_0$ such that for every $x\geq x_0$, the interval $[x-x^\theta, x]$ contains a prime.
\end{lemma}

Fix $d$. 
Set $m= \frac{\log A}{1-\theta}$ and $r=\sum_{i=d}^m {m\choose i}$, where $\theta=0.525$. 
Lemma~\ref{prime} then provides a prime $q$ such that $r-r^{\theta}\leq q \leq r$. 
Take an arbitrary injective map $\psi$ from $\F_q$ to $\F_2^m$ such that $w_H{(\psi(x))}\geq d$ for all $x\in \F_q$.
In other words, $\psi$ maps symbols in $\F_q$ to binary sequences of length $m$ with weight at least $d$.
For a vector $\bx=x_0x_1\cdots x_{n-1}\in \F_q^n$, let $\psi(\bx)=\psi(x_0)\psi(x_1)\cdots \psi(x_{n-1})$.

Suppose that $n_2$ is divisible by $m$. 
Set $n_R\triangleq n_1({n_2}/{m})-4$ and $k_R\triangleq n_R-2(d+7)$. 
Then we have $k_R <n_R< q$ and set $M\triangleq q^{k_R/2}$. 
Now we provide our construction of robust positioning arrays.

%\begin{construction}\label{Construction-RPA}
\vspace{2mm}
\noindent{\bf Construction 2}. Let $\cG=(\bsg_0,\bsg_1,\ldots, \bsg_{M-1})$ be a $(k_R/2,q)$-Gray code and 
consider a Reed-Solomon code of length $n_R$ and dimension $k_R$  over $\F_q$.  
For each $0\leq i,j\leq M-1$, set $\bc_{ij}=\enc^{(n_R,k_R)}(\bsg_i\bsg_j)$. %be the unique codeword in $\cC$ such that $\bc_{ij}[0,k_R-1]=\bsg_i\bsg_j$.

Let $\bp=0^{k}\bu$ where $k=4m-\ell$ and $\bu$ is the $d$-auto-cyclic vector provided in \eqref{dautoc}.  
For each $\bc_{ij}$, the concatenation $\bp\psi(\bc_{ij})$ has length $n_1n_2$  as 
$n_R=n_1({n_2}/{m})-4$. Then let $\bA_{ij}$ be the  $n_1\times n_2$ array 
whose rows can be concatenated to form $\bp\psi(\bc_{ij})$ (see Fig.~\ref{fig-arraycon}).

Finally, construct a large array $\bA$ as
$$
\bA=
\left(
\begin{array}{cccc}
\bA_{00} & \bA_{01} & \cdots & \bA_{0, M-1}\\
\bA_{10} & \bA_{11} & \cdots & \bA_{1, M-1}\\
\vdots & \vdots & \ddots & \vdots\\
\bA_{M-1,0} & \bA_{M-1,1} & \cdots & \bA_{M-1, M-1}
\end{array}
\right).
$$

\begin{figure*}[htbp]
  \centering
  \includegraphics[width=10cm, trim={0 0 0 0}, clip]{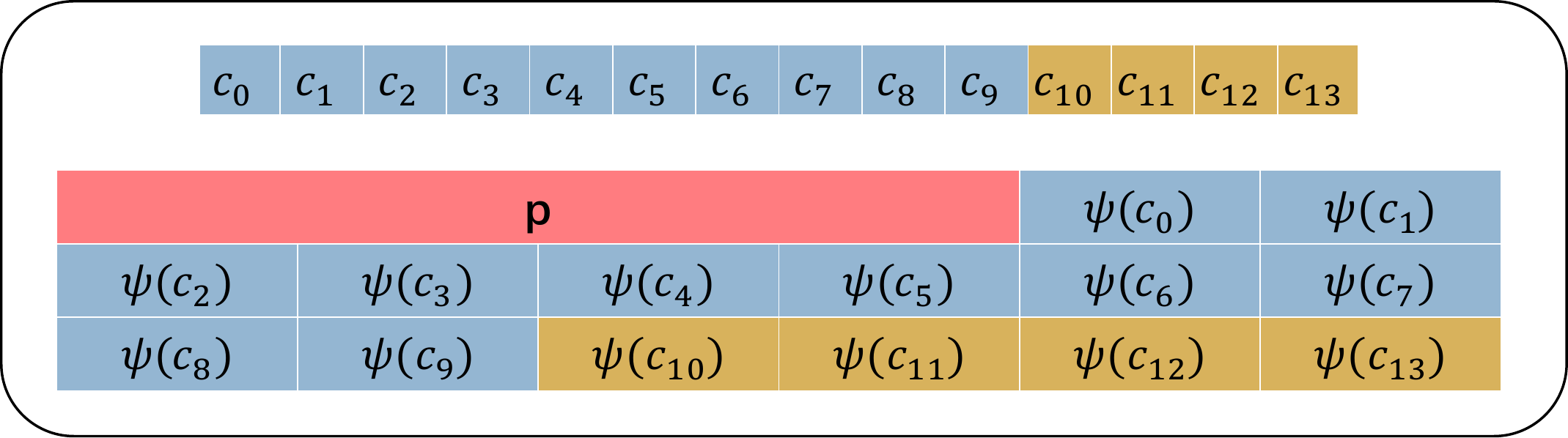}
  \caption{A codeword from a Reed-Solomon code of length $14$ and dimension $10$ and its corresponding $n_1\times n_2$ array with $n_1=3$  and $n_2=6m$. The blue cells represent the message bits and the yellow cells represent the check bits.}
  \label{fig-arraycon}
\end{figure*}

%Notice that in the construction above, each  $\bA_{ij}$ is composed of $\frac{n_1n_2}{m}-4$ vectors of length $m$ and the vector $\bp$. All of these vectors except $\bp$ have Hamming weights at least $d$, see Fig~\ref{fig-arraycon}. 

{For each $\bA_{ij}$, we refer to the zeros and ones in $\psi(\bc_{ij}[\ell])$ with $\ell<k_R$ as message bits and refer to those in $\psi(\bc_{ij}[\ell])$ with $\ell \geq k_R$ as check bits, see Fig.~\ref{fig-arraycon}.} 

For an array $\bM=(m_{i,j})$, we use $\bM[i_0,i_0+a-1][j_0,j_0+b-1]$ to denote the $a\times b$ cyclical subarray of $\bM$ whose top-left cell is $m_{i_0j_0}$. 
The following result is an analogue to Lemma~\ref{pcompare} and helps to locate the modular position efficiently.

\begin{lemma}\label{Lem:Comparison-Array}
Consider the subarray $\bW=\bA[i_0,i_0+n_1-1][j_0,j_0+n_2-1]$ in $\bA$.
Pick $i\in \bbracket{n_1}$ and $j\in\bbracket{n_2}$. Then the following hold.
\begin{enumerate}[(i)]
\item If $i+i_0\equiv 0 \ppmod{n_1}$ and $j+j_0\equiv 0 \pmod{n_2}$, then $\bW[i,i][j,j+4m-1]=\bp$.
\item If  $i+i_0\not \equiv 0 \ppmod{n_1}$ or $j+j_0\not \equiv 0 \ppmod{n_2}$, then $d_H(\bW[i,i][j,j+4m-1],\bp)\geq d$.
\end{enumerate}
%For a subarray $\bW=\bA[i_0,i_0+n_1-1][j_0,j_0+n_2-1]$ in $\bA$ and every $i\in [n_1]$ and $j\in[n_2]$, we have either $\bW[i,i][j,j+4m-1]=\bp$ if $i+i_0\equiv 0 \ppmod{n_1}$ and $j+j_0\pmod{n_2}$, or $d_H(\bW[i,i][j,j+4m-1],\bp)\geq d$  if  $i+i_0\not \equiv 0 \ppmod{n_1}$ or $j+j_0\not \equiv 0 \ppmod{n_2}$.
\end{lemma}

\begin{proof}
Let $\hai \in [n_1]$  and $\haj \in [n_2]$ such that $\hai+i_0\equiv 0\ppmod{n_1}$ and $\haj+j_0\equiv 0\ppmod{n_2}$. We consider the array $\bV$, which is obtained by shifting $\bW$ cyclically upwards $\hai$ times and leftwards $\haj$ times. Then $\bV[0,0][0,4m-1]=\bp$, see Fig.~\ref{Fig:Comparison-Array}, and it suffices to show  $d_H(\bV[i,i][j,j+4m-1],\bp)\geq d$  when $i\in[1,n_1-1]$ or $j\in[1,n_2-1]$.

For $i\in[1,n_1-1]$, since $k=4m-\ell >3m$, $\bV[i,i][j,j+k-1]$ must contain a length-$m$ vector $\psi(x_0)$ for some $x_0\in \F_q$ (see Fig.~\ref{Fig:Comparison-Array}). Observe that $\psi(x_0)$ has weight at least $d$.
Hence, we have 
\begin{align*}
d_H(\bV[i,i][j,j+4m-1],\bp) &\geq d_H(\bV[i,i][j,j+k-1],\bp[0,k-1]) \\
& =d_H(\bV[i,i][j,j+k-1],0^k)\geq d.
\end{align*}

For $i=0$ and $j\in[1,n_2-1]$, the proof follows from Lemma~\ref{pcompare}. 
\end{proof}

\begin{figure*}[htbp]
  \centering
  \includegraphics[width=12cm, trim={0 0 0 0}, clip]{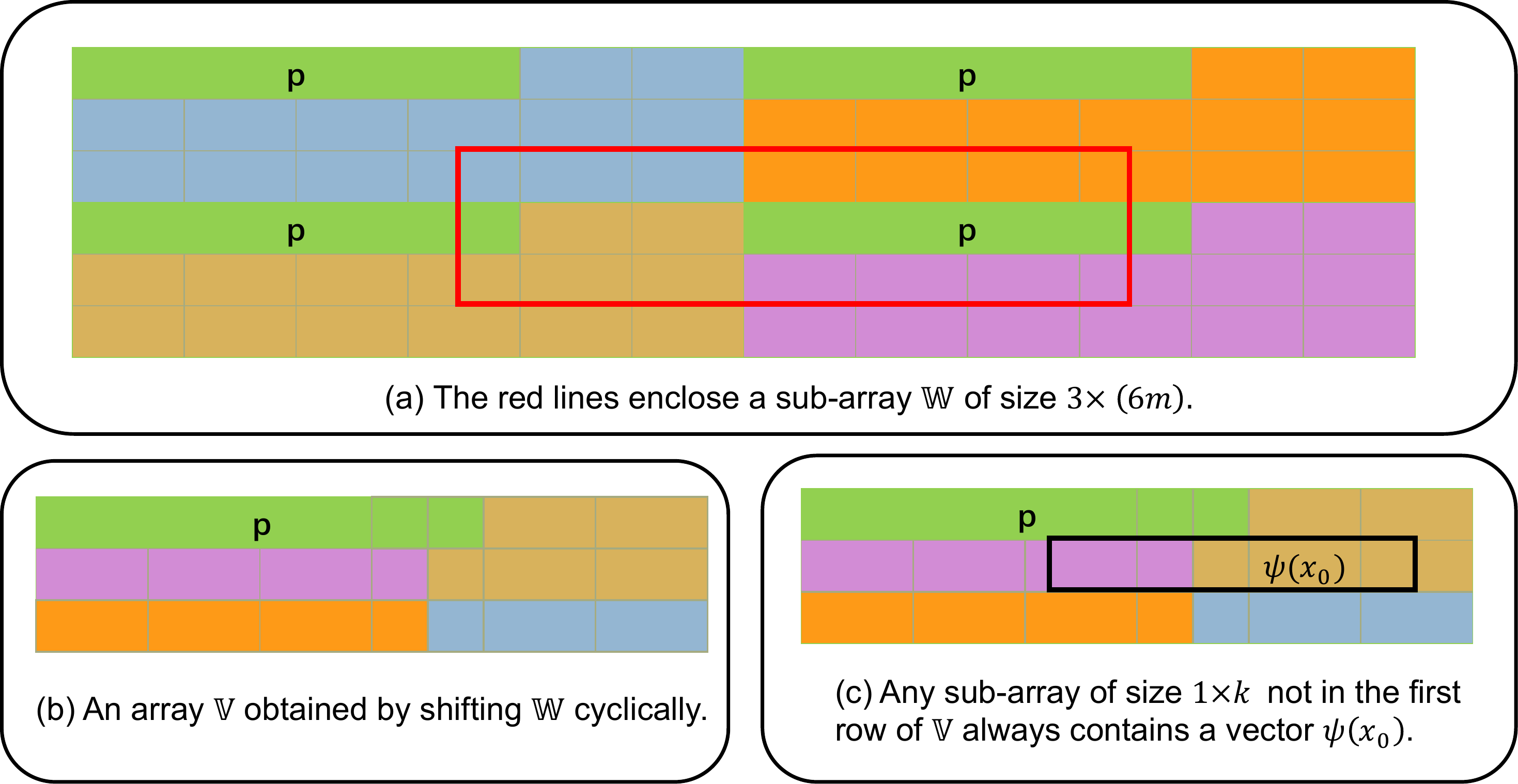}
  \caption{An example with $n_1=3$ and $n_2=6m$ to illustrate the proof of Lemma~\ref{Lem:Comparison-Array}. {\bf(a)} The red lines enclose the subarray $\bW=\bA[2,5][3.5m,9.5m-1]$. {\bf(b)} Shifting $\bW$ upwards one time and leftwards $2.5m$ times, we got the array $\bV$. {\bf(c)} Any subarray $\bV[i,i][j,j+k-1]$ with $i\not=0$ always contains a vector $\psi(x_0)$ for some $x_0\in \F_q$.}
  \label{Fig:Comparison-Array}
\end{figure*}

We regard an array of dimension  $a\times (bm)$ as a $a \times b$ partitioned matrix with each block being a vector of length $m$. 
Given a pair of $a \times (bm)$ arrays $\bM_1$ and $\bM_2$, 
we denote the Hamming distance of their corresponding partitioned matrices as $d_{B}(\bM_1,\bM_2)$. 
In other words, $d_{B}(\bM_1,\bM_2)$  counts the number of different blocks in $\bM_1$ and $\bM_2$. 
Therefore,  in Construction 2, %~\ref{Construction-RPA}, 
we have \[d_{H}(\bA_{ij},\bA_{i'j'}) \geq d_{B}(\bA_{ij},\bA_{i'j'})=d_H(\bc_{ij},\bc_{i'j'}).\]
For a pair of $a\times b$ arrays $\bM_1$ and $\bM_2$ with $b$ not divisible by $m$, 
we can repeat the last columns to form two arrays $\bM_1'$ and $\bM_2'$ of dimension   $a\times (\ceil{{b}/{m}}m)$. 
Then denote $d_B(\bM_1,\bM_2):=d_B(\bM_1',\bM_2')$.
Hence, $d_{B}(\bM_1,\bM_2)$ counts the number of different (truncated) blocks in $\bM_1$ and $\bM_2$. 

\begin{lemma}
\label{Lem:ModularComparision-Array}
For any two subarrays $\bW=\bA[i,i+n_1-1][j,j+n_2-1]$ and $\bW'=\bA[i',i'+n_1-1][j',j'+n_2-1]$ with $i\equiv i'\pmod{n_1}$ and $j\equiv j'\pmod{n_2}$, the Hamming distance between them is at least $d$.
\end{lemma}

\begin{proof}
Suppose that $i=an+\bai$ and $i'=a'n+\bai$ for some $\bai \in \bbracket{n_1}$,  and $j=bn+\baj$ and $j'=b'n+\baj$ for some $\baj \in \bbracket{n_2}$. Let $\hai\in\bbracket{n_1}$ and $\haj\in\bbracket{n_2}$ be the integers such that $\bai+\hai\equiv 0\pmod{n_1}$ and $\baj+\haj\equiv 0\pmod{n_2}$. Shift $\bW$ cyclically upwards $\hai$ times and leftwards $\haj$ times and denote the resulting array as $\bV$. Similarly, let $\bV'$ be the corresponding shifted array of $\bW'$. Then $d_H(\bW,\bW')=d_H(\bV,\bV')$. 
{ Since  thickness is bounded by a constant, we have   $\log n_1 / \log n_2 =O(1)$, and then  $(n_R-k_R)m < n_2$.} It follows that  the check bits of $\bA_{\alpha \beta}$ appear in the last row 
$\bA_{\alpha\beta}[n_1-1,n_1-1][n_2-(n_R-k_R)m,n_2-1]$.  To estimate $d_H(\bV,\bV')$, we proceed in three cases,   depending on where the check bits of $\bV$ and $\bV'$ come from.

\vspace{2mm}
\noindent{\bf Case 1}: $\baj\in [0,n_2-(n_R-k_R)m-1]$. As shown in Fig.~\ref{Fig:ModularComparision-Array}(a), 
we partition $\bW$ into four blocks by setting
\begin{align*}
\bW_\textup{I}  =\bW[0,\hai-1][0,\haj-1], & \ \bW_\textup{II}=\bW[0,\hai-1][\haj,n_2-1], \\
 \bW_\textup{III}=\bW[\hai,n_1-1][0,\haj-1], &  \ \bW_\textup{IV}=\bW[\hai,n_1-1][\haj,n_2-1],
\end{align*}
and  partition $\bA_{ab}$ into four blocks by setting
\begin{align*}
\bA_\textup{I}=\bA_{ab}[\bai,n_1-1][\baj,n_2-1], 
& \ \bA_\textup{II}=\bA_{ab}[\bai,n_1-1][0,\baj-1], \\
 \bA_\textup{III}=\bA_{ab}[0,\bai-1][\baj, n_2-1],  
& \  \bA_\textup{IV}=\bA_{ab}[0,\bai-1][0,\baj-1].
\end{align*}
Then  we have 
$$\bV=
\left(\begin{array}{cc}
\bW_\textup{IV} & \bW_\textup{III}\\
\bW_\textup{II} & \bW_\textup{I}
\end{array}
\right),
\bA_{ab}=
\left(\begin{array}{cc}
\bA_\textup{IV} & \bA_\textup{III}\\
\bA_\textup{II} & \bA_\textup{I}
\end{array}
\right),
$$
and $\bW_\textup{I}=\bA_\textup{I}.$
%We propose to estimate $d_B(\bV,\bA_{ab})$. If $\baj$ is not divisible by $m$, we could repeat the last columns of $\bW_\textup{II}, \bW_\textup{IV},\bA_\textup{II}$, and $\bA_\textup{IV}$ or the first columns of $\bW_\textup{I}, \bW_\textup{III}, \bA_\textup{I}$, and $\bA_\textup{III}$, so that their widths are divisible by $m$.

Notice that all these blocks, except $\bW_\textup{I}$ and $\bA_\textup{I}$, do not contain check bits, and $(\bsg_0,\bsg_1\ldots,\bsg_{M-1})$ is a Gray code.
It follows that 
\begin{align*}
d_B(\bW_\textup{II},\bA_\textup{II}) \leq  d_H(\bsg_a\bsg_b, \bsg_{a}\bsg_{b+1})\leq 1, \\
d_B(\bW_\textup{III},\bA_\textup{III})\leq d_H(\bsg_a\bsg_b, \bsg_{a+1}\bsg_{b})\leq 1,\\ 
\textup{and \ } d_B(\bW_\textup{IV},\bA_\textup{IV})\leq  d_H(\bsg_a\bsg_b, \bsg_{a+1}\bsg_{b+1})\leq 2.
\end{align*}
So, 
\[
d_B(\bV,\bA_{ab})\leq d_B(\bW_\textup{I},\bA_\textup{I}) +  d_B(\bW_\textup{II},\bA_\textup{II}) + d_B(\bW_\textup{III},\bA_\textup{III}) +d_B(\bW_\textup{IV},\bA_\textup{IV}) \leq 4.
\]
With the same argument, we can get $d_B(\bV',\bA_{a'b'})\leq 4$. Hence,
\[
d_H(\bV,\bV')\geq d_B(\bV,\bV')\geq  d_B(\bA_{ab},\bA_{a'b'})-8=d_H(\bc_{ab},\bc_{a'b'})-8 \geq d.
\]

\vspace{2mm}
\noindent{\bf Case 2}: $\baj\in [n_2-(n_R-k_R)m, n_2-(n_R-k_R)m/2 -1]$. 
%For each $\bA_{\alpha \beta}$,
For $\alpha,\beta\in\bbracket{M}$, we change the bits in the block $\bA_{\alpha \beta}[n_1-1,n_1-1][n_2-(n_R-k_R)m,n_2-(n_R-k_R)m/2-1]$ to one and 
denote the resulting array as $\bar{\bA}_{\alpha \beta}$. Let $\bar{\bc}_{\alpha \beta}$ be the shortened codeword of $\bc_{\alpha,\beta}$ by deleting the subword $\bc_{\alpha \beta}[k_R, n_R-(n_R-k_R)/2-1]$
(see Fig.~\ref{Fig:ModularComparision-Array}(b)). 
Then we have 
\begin{equation}
\label{eq:ModularComparision-Array-1}
d_B(\bar{\bA}_{\alpha\beta},\bar{\bA}_{\alpha'\beta'})=d_H(\bar{\bc}_{\alpha\beta}, \bar{\bc}_{\alpha'\beta'}) \geq  (n_R-(n_R-k_R)/2)-k_R+1\geq d+8.
\end{equation}
Now, let $\bar{\bW}$, $\bar{\bW}'$, $\bar{\bV}$ and  $\bar{\bV}'$ be the corresponding arrays of ${\bW}$, ${\bW}'$, ${\bV}$ and  ${\bV}'$ with some check bits being changed to one. 
As in Case 1, we can show that 
\begin{equation}
\label{eq:ModularComparision-Array-2}
d_B(\bar{\bV},\bar{\bA}_{ab})\leq 4 \textup{\ and\ } d_B(\bar{\bV}',\bar{\bA}_{ab}')\leq 4.
\end{equation}
The only difference is that $\bar{\bW}_\textup{II}$, $\bar{\bA}_\textup{II}$,     $\bar{\bW}'_\textup{II}$, and  $\bar{\bA}'_\textup{II}$   may contain the check bits. %However, we already changed these bits to one, 
However, these bits are set to one, so we have $d_B(\bar{\bW}_\textup{II},\bar{\bA}_\textup{II})\leq  d_H(\bsg_a\bsg_b, \bsg_{a}\bsg_{b+1})\leq   1$, see Fig.~\ref{Fig:ModularComparision-Array}(c).

%According to (\ref{eq:ModularComparision-Array-1}) and (\ref{eq:ModularComparision-Array-2}), it follows that  
It follows from \eqref{eq:ModularComparision-Array-1} and \eqref{eq:ModularComparision-Array-2} that
$$d_B(\bar{\bV},\bar{\bV}')\geq d_B(\bar{\bA}_{ab},\bar{\bA}_{a'b'})-8\geq d,$$ and then 
$$d_H(\bV,\bV')\geq d_H(\bar{\bV},\bar{\bV}')\geq d_B(\bar{\bV},\bar{\bV}')\geq   d.$$

\vspace{2mm}
\noindent{\bf Case 3}: $\baj\in [n_2-(n_R-k_R)m/2, n_2-1]$.  %For each $\bA_{\alpha \beta}$, 
For $\alpha,\beta\in\bbracket{M}$,
we change the bits in the block $\bA_{\alpha \beta}[n_1-1,n_1-1][n_2-(n_R-k_R)m/2,n_2-1]$ to one and 
denote the resulting array as $\tilde{\bA}_{\alpha \beta}$. 
Let $\tilde{\bc}_{\alpha \beta}$ be the shortened codeword of $\bc_{\alpha,\beta}$ by deleting the last $(n_R-k_R)/2$ bits.

Now, let $\tilde{\bW}$, $\tilde{\bW}'$, $\tilde{\bV}$ and  $\tilde{\bV}'$ be the corresponding arrays of ${\bW}$, ${\bW}'$, ${\bV}$ and  ${\bV}'$. Again we use the strategy in Case 1 to show that 
\[
d_H(\bV,\bV')\geq d_H(\tilde{\bV},\tilde{\bV}') \geq d_B(\tilde{\bV},\tilde{\bV}')\geq d_B(\tilde{\bA}_{a,b+1},\tilde{\bA}_{a',b'+1})-8 \geq d.
\]
We note that in this case we need to partition $\tilde\bA_{a,b+1}$ instead of  $\tilde\bA_{ab}$, see Fig.~\ref{Fig:ModularComparision-Array}(d). 
 \end{proof}

\begin{figure*}[htbp]
  \centering
  \includegraphics[width=12cm, trim={0 0 0 0}, clip]{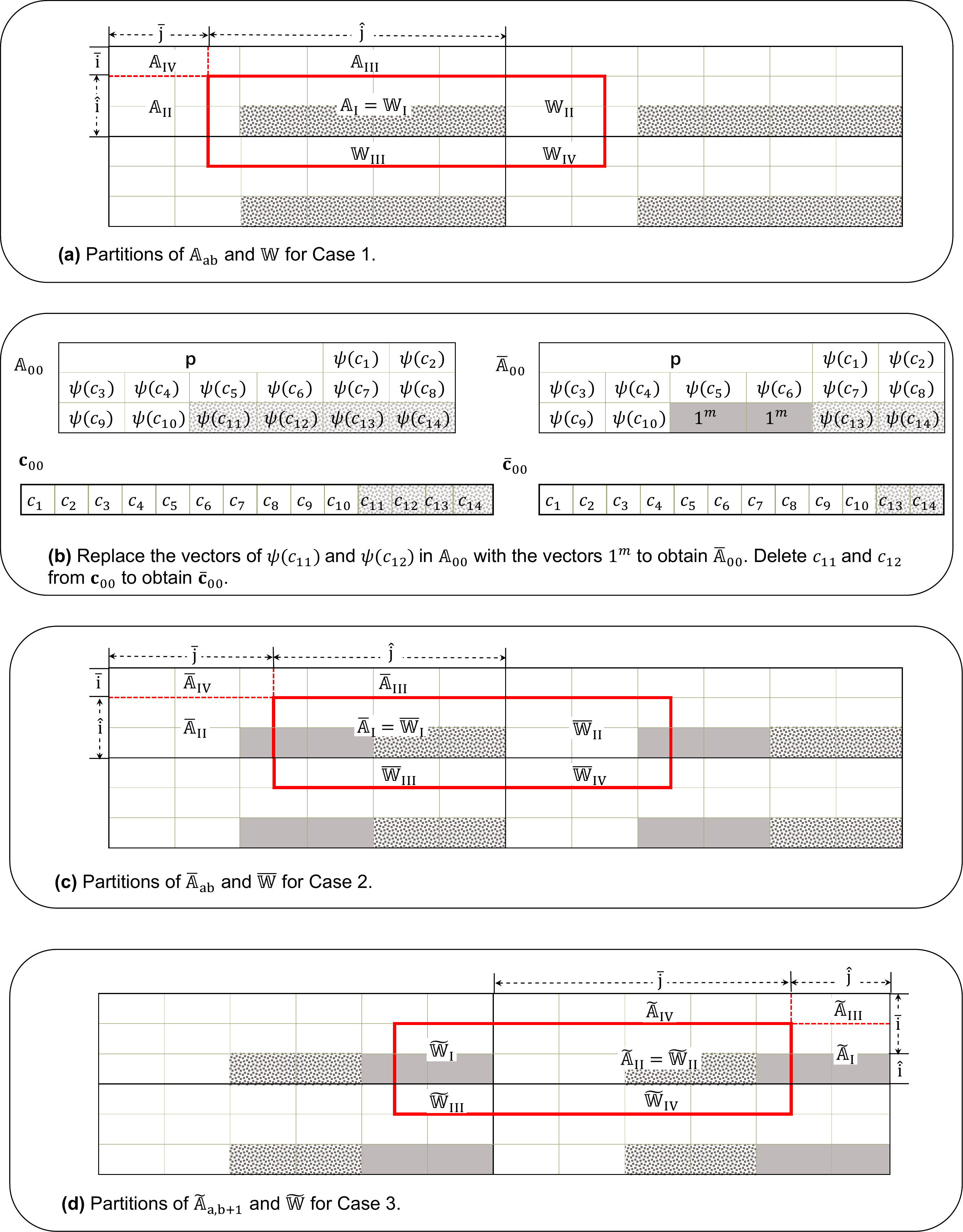}
  \caption{An example with $n_1=3$ and $n_2=6m$ to illustrate the proof of Lemma~\ref{Lem:ModularComparision-Array}. The black lines enclose the arrays $\bA_{ab}$,    $\bA_{a,b+1}$,  $\bA_{a+1,b}$ and    $\bA_{a+1,b+1}$.  The red lines enclose the subarray $\bW$. The empty blocks represent the vectors of message bits. The blocks with dots represent the vectors of check bits. The solid blocks represent the vectors $1^m$. }
  \label{Fig:ModularComparision-Array}
\end{figure*}

\begin{theorem}
%The array $\bA$ in Construction~\ref{Construction-RPA} is an $(n_1\times n_2,d)$-RPA.
The array $\bA$ in Construction 2 is an $(n_1\times n_2,d)$-RPA.
\end{theorem}

\begin{proof}  Let $\bW$ and $\bW'$ be two distinct subarrays of dimension  $n_1\times n_2$ in $\bA$.
Assume that $\bW=\bA[i,i+n_1-1][j,j+n_2-1]$ and $\bW'=\bA[i',i'+n_1-1][j',j'+n_2-1]$, where $(i,j)\not=(i',j')$.  
From Lemma~\ref{Lem:ModularComparision-Array},  we have $d_H(\bW,\bW')\geq d$  when $i\equiv i' \pmod{n_1}$ and $j\equiv j' \pmod{n_2}$. Now we consider the case where $i\not\equiv i' \pmod{n_1}$ or $j\not\equiv j' \pmod{n_2}$. Let $\hai \in [n_1]$ and $\haj\in [n_2]$ such that $i+\hai\equiv 0 \ppmod{n_1}$ and $j+\haj\equiv 0 \ppmod{n_2}$. So we have $i'+\hai\not \equiv 0 \ppmod{n_1}$ or $j'+\haj\not \equiv 0 \ppmod{n_2}$. It follows from  Lemma~\ref{Lem:Comparison-Array}  that   $\bW[\hai,\hai][\haj,\haj+4m-1]=\bp$  and $d_H(\bW'[\hai,\hai][\haj,\haj+4m-1],\bp)\geq d$. Hence
\[
d_H(\bW,\bW')\geq d_H(\bW[\hai,\hai][\haj,\haj+4m-1], \bW'[\hai,\hai][\haj,\haj+4m-1])\geq d,
\]
which completes the proof.
\end{proof}

\begin{corollary}\label{redundancy-array}
Let $\bW$ be a  window of area $A$ and  thickness bounded by a constant. Then there is a binary RPA  for $\bW$ with distance $d$ and  redundancy at most $$4.21 d\log{A} + 35.79\log{A}+o(1),$$ 
provided $A$ is large enough. 
%For sufficient large $n_1n_2$ with $\log n_1 / \log n_2 = O(1)$, 
\end{corollary}

\begin{proof} 
%The RPA $\bA$ in Construction~\ref{Construction-RPA} has size $(n_1M) \times (n_2M)$, where  $M= q^{k_R/2}$.  
The RPA $\bA$ in Construction 2 has dimension $(n_1M) \times (n_2M)$, where  $M= q^{k_R/2}$.
So, its redundancy is given by
\[
%\label{eq:redundancy-array1}
n_1n_2-\log (n_1n_2M^2) = A -k_R\log q -\log A,
\]
where $k_R=n_R-2(d+7)= A/m-2d-18,$
%\begin{align}
%\label{eq:redundancy-array2}
%k_R=n_R-2(d+7)=\frac{n_1n_2}{m}-2d-18,
%\end{align}
and 
\begin{equation}\label{eq:redundancy-array3}
\log  q  \geq \log (r-r^{\theta}) =\log r + \log \left(1-\frac{1}{r^{1-\theta}}\right)  \geq \log r - \frac{\log e}{r^{1-\theta}-1}.
\end{equation}

Recall that $\theta=0.525$, and so, $m={\log A}/(1-\theta) \geq ({3}/{2})\log A.$ It follows that 
\begin{align*}
r&= \sum_{i=d}^m {m\choose i}=2^m\left[1-\sum_{i=0}^{d-1} {m\choose i}\frac{1}{2^m}\right] \geq 2^m\left[1-\exp\left(-\frac{2(m/2-d+1)^2}{m}\right) \right]\\
& \geq 2^m\left(1-e^{-\frac{2m}{3 \log e}} \right) \geq 2^m\left(1-\frac{1}{A} \right).
\end{align*}
Then we have
\begin{align}
\label{eq:redundancy-array4}
\log r \geq m+\log\left(1-\frac{1}{A} \right)   \geq m-\frac{\log e}{A-1}.
\end{align}
On the other hand,
\begin{align}
\label{eq:redundancy-array5}
r^{1-\theta}\geq \left(\frac{2^m}{2}\right)^{1-\theta}=\frac{A}{2^{1-\theta}}.
\end{align}
Combining (\ref{eq:redundancy-array3}),  (\ref{eq:redundancy-array4}) and (\ref{eq:redundancy-array5}), we get
\begin{align*}
\log q \geq m - O\left(\frac{1}{A}\right).
\end{align*}
Hence, the redundancy of $\bA$ is at most 
\begin{align*}
A\ -\ &\left(\frac{A}{m}-2d-18\right)\left(m - O\left(\frac{1}{A}\right)\right) -\log A\\
=& \ \frac{2d}{1-\theta} \log A+\frac{18}{1-\theta}\log A  - \log A+o(1)  \\
\approx& \ 4.21 d\log A+ {36.89}\log A+o(1).
\end{align*}
\end{proof}

%\subsection{Locating Algorithm}
%In this subsection, we give an efficient locating algorithm for the array $\bA$ in Construction~\ref{Construction-RPA}.
To conclude, we provide an efficient locating algorithm for the array $\bA$ in Construction~2.
Let $\chi$ be a map from $\F_2^m$ to $\F_q$ such that
$\chi(\psi(x))=x$ for all  $x\in \F_q$  and $\chi (\bv)=0$ for all $\bv\not \in \{\psi(x): x\in \F_q\}$.

We briefly describe Algorithm~\ref{locatingalgo2}.
Suppose that $\bW$ is an $n_1\times n_2$ subarray of $\bA$ that is corrupted at $e$ positions with $e\le (d-1)/2$.
So  there is a unique pair $(i,j)$ such that $d_H(\bA[i,i+n-1][j,j+n-1],\bW)\le (d-1)/2$. Assume that $i=an_1+\bai$ and $j=bn_2+\baj$ with $\bai\equiv i\pmod{n_1}$ and $\baj\equiv j \pmod{n_2}$.
In what follows, we briefly describe how to determine $a,b,\bai$ and $\baj$.

\begin{enumerate}[(I)]
\item We first use  Lemma~\ref{Lem:Comparison-Array} to determine $\bai$ and $\baj$. 
\item Next, we rotate $\bW$  appropriately to obtain $\bV$, so that the concatenation of the rows of $\bV$, denoted as $\bv$, is the binary image obtained from either a $q$-ary codeword $\bc_{a,b}$ or a concatenation of some shortened codewords $\bc_{a,b}$, $\bc_{a+1,b}$, $\bc_{a,b+1}$, and $\bc_{a+1,b+1}$. Since $\bv$ is obtained via the map $\psi$ and prepending the string $\bp$, we reverse this process to obtain the $q$-ary estimate $\bu$. 
\item Finally, depending on the value of $\baj$, we apply the Reed-Solomon decoding algorithm $\dec$ to find either $\bc_{a,b}$ (when  $\baj\in [0,n_2-2(d+7)m-1]$), some shortened version of $\bc_{a,b}$ (when $\baj \in [n_2-2(d+7)m, n_2-(d+7)m-1]$), or some shortened version of  $\bc_{a,b+1}$ (when $j\in [n_2-(d+7)m, n_2-1]$). Therefore, we determine $a$ and $b$ and hence, obtain $i=an_1+\hai$ and $j=bn_2+\haj$.   \end{enumerate}

The first step above requires $\ell_pn_1n_2$ comparisons.
The Reed-Solomon decoding  runs in $O(n_R^3)=O((n_1n_2)^3)$ time,
and the decoding of Gray codes runs in $O(k_R(\log q)^2)=O(n_1n_2(\log (n_1n_2))^2)$ time.
Therefore, Algorithm~\ref{locatingalgo} can determine the location  in $O((n_1n_2)^3)$, or equivalently  $O(A^3)$ time.

%The time complexity of the algorithm mostly comes from the decoding of Reed-Solomon codes, which is $O((n_1n_2)^3)$.

\begin{algorithm}[h]
\begin{algorithmic}
\caption{Locating algorithm for the array $\bA$ in Construction 2} %~\ref{Construction-RPA}}
\label{locatingalgo2}

\STATE {\bf Input}: an $n_1\times n_2$ window $\bW$ of area  $A$ and thickness bounded by a constant %in the $(n,d)$-RPA $\bcs$ %$n$, $d$, $e$, and  
\STATE {\bf Output}: a position $(i,j)\triangleq(an_1+\bai,bn_2+\baj)$ 
such that $d_H(\bA[i,i+n_1-1][j,j+n_2-1],\bW)\le (d-1)/2$
%\STATE {\bf Input}: $n_1$, $n_2$, $d$, $e$, and  an array $\bW$ of size  $n_1\times n_2$
%\STATE {\bf Output}: a position $(an_1+\bai, bn_2+\baj)$
\STATE{}
\STATE $m \leftarrow \frac{\log A}{1-\theta}$
\STATE $n_R \leftarrow n_1(n_2/m)-4$
\STATE $k_R \leftarrow n_R-2(d+7)$
\STATE{}
\STATE $(\hai,\haj) \leftarrow$ unique tuple such that $d_H(\bW[\hai,\hai][\haj,\haj + 4m -1],\bp)\le (d-1)/2$
\STATE Set $\bai \in \bbracket{n_1}$ such that $\bai+\hai\equiv 0\ppmod{n_1}$
\STATE Set $\baj \in \bbracket{n_2}$ such that $\baj+\haj\equiv 0\ppmod{n_2}$
\STATE{}
%\FORALL {$i \in [n_1]$ and $j \in [n_2]$}
%	\IF {$d_H(\bW[i,i][j,j+4m -1],\bp) < e$}
%	\STATE $\hai \leftarrow i$, $\haj \leftarrow j$
%	\ENDIF
%\ENDFOR

\STATE $\bV \leftarrow$ the array obtained by shifting $\bW$ cyclically upwards $\hai$ times and leftwards $\haj$ times
\STATE $\bv \leftarrow $ the concatenation of the rows of $\bV$
\STATE $\bu \leftarrow \chi(\bv[4m,5m-1])\chi(\bv[5m,6m-1])\chi(\bv[6m,7m-1])\cdots\chi(\bv[n_1n_2-m,n_1n_2-1]) $\\
\STATE{}
%\STATE Let $\bai \in [n_1]$ and $\baj\in [n_2]$ such that $\bai+\hai\equiv 0\ppmod{n_1}$ and  $\baj+\haj\equiv 0\ppmod{n_2}$ 
\IF {$\baj \in [0,n_2-2(d+7)m-1]$}
\STATE $\bc\gets\dec^{(n_R,k_R)}(\bu)$
%\STATE Run the decoding algorithm for Reed-Solomon codes of  length $n_R$ and dimension $k_R$ on $\bu$ to get a vector $\bc$
%\STATE Run the encoding algorithm for Gray codes on $\bc[0,k_R/2-1]$ and $\bc[k_R/2,k_R-1]$ to get indices $a$ and $b$
\ELSIF {$\baj \in [n_2-2(d+7)m, n_2-(d+7)m -1]$}
\STATE $\bu^s \leftarrow$ the shortened codeword $\bu[0,k_R-1]\bu[k_R+(d+7),n_R-1]$
\STATE $\bc\gets\dec^{(n_R-(d+7),k_R)}(\bu^s)$

%\STATE Run the decoding algorithm for Reed-Solomon codes of length $n_R-(d+7)$ and  dimension $k_R$ on $\bu^s$ to get a vector $\bc$
%\STATE Run the encoding algorithm for Gray codes  on $\bc[0,k_R/2-1]$ and $\bc[k_R/2,k_R-1]$ to get indices $a$ and $b$
\ELSE
\STATE $\bu^s \leftarrow$ the shortened codeword $\bav[0,k_R+d+6]$
\STATE $\bc\gets\dec^{(n_R-(d+7),k_R)}(\bu^s)$
%\STATE Run the decoding algorithm for Reed-Solomon codes of length $n_R-(d+7)$ and  dimension $k_R$ on $\bu^s$ to get a vector $\bc$
%\STATE Run the encoding algorithm for Gray codes  on $\bc[0,k_R/2-1]$ and $\bc[k_R/2,k_R-1]$ to get indices $a$ and $b+1$
\ENDIF
\STATE $a\gets \decG(\bc[0,k_R/2-1])$
{
\IF {$\baj \in [0, n_2-(d+7)m -1]$}
\STATE $b\gets \decG(\bc[k_R/2,k_R-1])$
\ELSE
\STATE $b+1\gets \decG(\bc[k_R/2,k_R-1])$
\ENDIF
}
\RETURN {$(an_1+\bai, bn_2+\baj)$} 
\end{algorithmic}
\end{algorithm}

%\newpage

\section{Binary Positioning Arrays with Constant Rank Distance}
\label{Sect:biRPAconstantrd}
%\subsection{Construction}
%In this subsection, 

We continue our investigation of binary robust positioning arrays.
In this section, we consider the scenario where the error patterns are confined to a certain number of rows or columns (or both). 
To correct for such errors, Roth demonstrated that it suffices to consider codes in the {\em rank distance} metric \cite{roth1991maximum}.

For two matrices $\bM_1$ and $\bM_2$ of the same dimension, the {\it rank distance} between them, denoted as $d_R(\bM_1,\bM_2)$, is defined as the rank of their difference, i.e., $d_R(\bM_1,\bM_2) \triangleq{\rm rank}(\bM_1-\bM_2)$. 
In this section, we modify Construction 2 %Construction~\ref{con-binary} 
to produce a binary  positioning array of strength $n_1\times n_2$ and rank distance $d$, i.e., a large array in which the rank distance between any two  $n_1\times n_2$ submatrices is  at least $d$.
Since a code $\cM \subseteq \F_q^{n_1\times n_2}$ with minimum rank distance $d$ satisfies the Singleton bound, i.e., $|\cM|\leq q^{n_2(n_1-d+1)}$, the redundancy of such an array should be at least $n_2(d-1)-O(1)$. 

To present our construction, we require the concept of maximum rank distance (MRD) codes.

\begin{theorem}[Maximum rank distance (MRD) code~\cite{Gabidulin:1985}]\label{thm:MRD}
Let $q$ be a prime power. Suppose that $N_1 \leq N_2$.
Then there exists a linear code $\cM \leq  F_q^{N_1\times N_2}$ of rank distance $d$ and dimension $N_2(N_1-d+1)$.\end{theorem}

We also need to choose a new marker $\bP$. Fix $d$ and let $m$ be an integer  such that  $m(m-d+1)=\log (n_1n_2)$.    Let $\bP$ be a  $4m\times 4m$ array in which
\begin{enumerate}[(i)]
\item the diagonal is $0^{4m-\ell} \bu$, where $\bu$ is the $d$-auto-cyclic vector provided in \eqref{dautoc};
\item the $d\times d$ subarrays at the right top corner and left bottom corner are identity matrices;
\item the symbols in all the other entries are $0$.
\end{enumerate}

%$$
%\bP=
%\left(
%\begin{array}{ccccccccc}
%0 & & & & & & 1 & & \\
%& \ddots & & & & & & \ddots & \\
%& & \ddots & & & & &  & 1\\
%& & & 0 & & & & &\\
%& & & & u_1 & & & &\\
%& & & & & u_2 & & &\\
%1 & & & & & & \ddots & &\\
%& \ddots & & & & & & \ddots & \\
%& & 1 & & & & & &  u_\ell \\
%\end{array}
%\right).
%$$

Let $\cM\subseteq \F_2^{m\times m}$ be an MRD code of rank distance $d$ and dimension $m(m-d+1)$. Suppose that both  $n_1$  and $n_2$ are divisible by $m$. 
Set $n_R\triangleq \frac{n_1n_1}{m^2}-16$ and $k_R\triangleq n_R-24$. Choose a prime power $q$ such that $n_R \leq  q <2^{m(m-d+1)}$ and set  $M\triangleq q^{k_R/2}$.  Take an arbitrary injective map $\psi$ from $\F_q$ to $\cM \backslash \{\mathbf{0}\}$, where $\mathbf{0}$ is the all-zero matrix.
So $\psi$ maps the elements of $\F_q$ to $m\times m$ matrices of rank at least $d$. 
%Now we provide our construction of robust positioning arrays.

%\begin{construction}\label{Construction-RPA}
\vspace{2mm}
\noindent{\bf Construction 3}. Let $\cG=(\bsg_0,\bsg_1,\ldots, \bsg_{M-1})$ be a $(k_R/2,q)$-Gray code and 
consider a Reed-Solomon code of length $n_R$ and dimension $k_R$  over $\F_q$.  
For each $0\leq i,j\leq M-1$, set $\bc_{ij}=\enc^{(n_R,k_R)}(\bsg_i\bsg_j)$. %be the unique codeword in $\cC$ such that $\bc_{ij}[0,k_R-1]=\bsg_i\bsg_j$.

For each $\bc_{ij}$, apply the map $\psi$ to the symbols of $\bc_{ij}$ to obtain $n_R$ $m\times m$ matrices of rank at least $d$. Since $n_R=(n_1n_2)/m^2-16$ and $n_R-k_R=24$, tile these $n_R$ matrices together with the marker $\bP$ to form an $n_1\times n_2$ array $\bA_{ij}$  (see Fig.~\ref{fig-arraycon-r})  such that 
\begin{enumerate}[(i)] 
\item $\bA_{ij}[0,4m-1][0,4m-1]=\bP$;
\item for each $\ell \in \{1,2,3\}$,  $\bA_{ij}[n_1-2m\ell,n_1-2m(\ell-1)-1][n_2-4m\ell,n_2-4m (\ell-1)-1]$ comprises eight $m\times m$ submatrices, each of  which corresponds to a check bit of $\bc_{ij}$.
\end{enumerate}

Finally, construct a large array $\bA$ as
$$
\bA=
\left(
\begin{array}{cccc}
\bA_{00} & \bA_{01} & \cdots & \bA_{0, M-1}\\
\bA_{10} & \bA_{11} & \cdots & \bA_{1, M-1}\\
\vdots & \vdots & \ddots & \vdots\\
\bA_{M-1,0} & \bA_{M-1,1} & \cdots & \bA_{M-1, M-1}
\end{array}
\right).
$$

\begin{figure*}[htbp]
  \centering
  \includegraphics[width=6cm, trim={0 0 0 0}, clip]{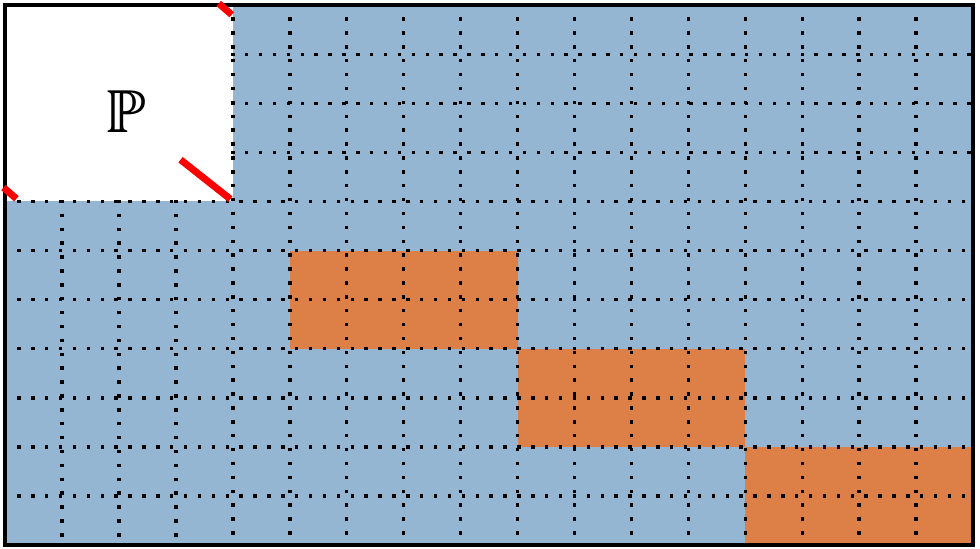}
  \caption{An example for the matrix $A_{ij}$ in Construction~3 with $n_1=11m$ and $n_2=17m$.  The blue cells represent the message bits and the yellow cells represent the check bits.}
  \label{fig-arraycon-r}
\end{figure*}

\begin{lemma}\label{Lem:Comparison-Array-r}
Consider the subarray $\bW=\bA[i_0,i_0+n_1-1][j_0,j_0+n_2-1]$ in $\bA$.
Pick $i\in \bbracket{n_1}$ and $j\in\bbracket{n_2}$. Then the following hold.
\begin{enumerate}[(i)]
\item If $i+i_0\equiv 0 \ppmod{n_1}$ and $j+j_0\pmod{n_2}$, then $\bW[i,i+4m-1][j,j+4m-1]=\bP$.
\item If  $i+i_0\not \equiv 0 \ppmod{n_1}$ or $j+j_0\not \equiv 0 \ppmod{n_2}$, then $d_R(\bW[i,i+4m-1][j,j+4m-1],\bP)\geq d$.
\end{enumerate}
%For a sub-array $\bW=\bA[i_0,i_0+n_1-1][j_0,j_0+n_2-1]$ in $\bA$ and every $i\in [n_1]$ and $j\in[n_2]$, we have either $\bW[i,i][j,j+4m-1]=\bp$ if $i+i_0\equiv 0 \ppmod{n_1}$ and $j+j_0\pmod{n_2}$, or $d_H(\bW[i,i][j,j+4m-1],\bp)\geq d$  if  $i+i_0\not \equiv 0 \ppmod{n_1}$ or $j+j_0\not \equiv 0 \ppmod{n_2}$.
\end{lemma}

\begin{proof} For simplicity, we assume that $n_1=n_2=n$. The case of $n_1\not=n_2$ can be proceeded similarly. 
Let $\hai,\haj \in \bbracket{n}$  such that $\hai+i_0\equiv 0\ppmod{n}$ and $\haj+j_0\equiv 0\ppmod{n}$. We consider the array $\bV$, which is obtained by shifting $\bW$ cyclically upwards $\hai$ times and leftwards  $\haj$ times. Then $\bV[0,4m-1][0,4m-1]=\bP$ and it suffices to show that $d_R(\bV[i,i+4m-1][j,j+4m-1],\bP)\geq d$  for $(i,j)\in \bbracket{n}^2\backslash\{(0,0)\}$. Write $\Diag(\bx)$ for a diagonal matrix whose diagonal is $\bx$.

We first assume  $i=j$. Similar to the proof of Lemma~\ref{pcompare},   we consider the following cases.

\vspace{2mm}
\noindent{\bf Case 1a}: $i\in [1, d]$. Then $\bP[4m-\ell-i, 4m-1-i][4m-\ell-i,4m-1-i]= \Diag(0^i\bu[0,{\ell-i}-1])$. On the other hand, the corresponding subarray in $\bV[i,i+4m-1][j,j+4m-1]$ is $\bV[4m-\ell,4m-1][4m-\ell,4m-1]$, which is equal to $\Diag(\bu)$. Due to the property of $\bu$, the rank distance between them is at least  $d$ and so  $d_R(\bV[i,i+4m-1][j,j+4m-1],\bP)\geq d$.

\vspace{2mm}
\noindent{\bf Case 1b}: $i\in [d+1, 4m-d]$. Since $4m-\ell > \ell$, $\bV[i,i+4m-\ell-1][i,i+4m-\ell-1]$ contains at least one  subarray which is equal to $\Diag(1^d)$. Note that  $\bP[0,4m-\ell-1][0,4m-\ell-1]=\mathbf{0}$. The rank distance between $\bV[i,i+4m-\ell-1][i,i+4m-\ell-1]$ and  $\bP[0,4m-\ell-1][0,4m-\ell-1]$ is at least $d$ and so  we have $d_R(\bV[i,i+4m-1][j,j+4m-1],\bP)\geq d$.

\vspace{2mm}
\noindent{\bf Case 1c}: $i\in [4m-d+1,n-(4m-\ell)]$. Since  $4m-\ell-d>3m$,   $\bV[i+d,i+4m-\ell-1][i+d,i+4m-\ell-1]$ should contain at least one $m\times m$ subarray of rank at least $d$.  Noting that $\bP[i+d,i+4m-\ell-1][i+d,i+4m-\ell-1]=\mathbf{0}$, again we have  $d_R(\bV[i,i+4m-1][j,j+4m-1],\bP)\geq d$.

\vspace{2mm}
\noindent{\bf Case 1d}: $i\in [n-(4m-\ell)+1, n-d]$. Since  $i+4m-\ell-n\geq 1$ and $i+4m-\ell+d-1-n\leq 4m-\ell-1$, we have
 $\bV[i+4m-\ell,i+4m-\ell+d-1][i+4m-\ell,i+4m-\ell+d-1]=\mathbf{0}$. Note that $\bP[4m-\ell,4m-\ell+d-1]=\Diag(1^d)$. It follows that  $d_R(\bV[i,i+4m-1][j,j+4m-1],\bP)\geq d$.

\vspace{2mm}
\noindent {\bf Case 1e}: $i \in [n-d+1,n-1]$.  Let $\delta=n-i$, then $\delta\in[1,d-1]$. We have
\begin{align*}
 & \bV[i+4m-\ell,i+4m-1][i+4m-\ell,i+4m-1] \\
 = & \bV[4m-\ell-\delta, 4m-1-\delta] [4m-\ell-\delta, 4m-1-\delta] \\
 = & \Diag( 0^\delta\bu[0,\ell-\delta-1]).
\end{align*}
Since $\bP[4m-\ell,4m-1][4m-\ell,4m-1]=\Diag(\bu)$, the rank distance between them is at least $d$ and so we have $d_R(\bV[i,i+4m-1][j,j+4m-1],\bP)\geq d$.

\vspace{2mm}

In the following we assume that $i<j$; the case of $i>j$ can be proceeded in the same way.

\vspace{2mm}
\noindent{\bf Case 2a}: $j\in [1, 4m-2d]$. Then  
$d \leq 4m-d-j < 4m-d-i$ and $4m-1-j<4m-1-i\leq 4m-1$. It follows that the subarray $\bP[4m-d-i,4m-1-i][4m-d-i,4m-1-i]$ is an upper triangular matrix with all entries on the diagonal being $0$. On the other hand, in $\bV[i,i+4m-1][j,j+4m-1]$ the corresponding subarray $\bV[4m-d,4m-1][4m-d,4m-1]$ is an identity matrix. Hence the rank distance between $\bV[i,i+4m-1][j,j+4m-1]$ and $\bP$ is at least $d$. 

\vspace{2mm}
\noindent{\bf Case 2b}: $j\in [4m-2d+1, n-(4m-\ell)]$. In this case we estimate the rank distance between  the submatrices $$\bV[i+2d,i+4m-\ell-1][j+2d,j+4m-\ell-1]$$ and $\bP[2d,4m-\ell-1][2d,4m-\ell-1]$. 
Since  $4m\leq j+2d$, $j+4m-\ell-1\leq n-1$ and $4m-\ell-2d >3m$, the subarray $\bV[i+2d,i+4m-\ell-1][j+2d,j+4m-\ell-1]$ always contains an $m\times m$ submatrix of rank at least $d$. On the other hand, $\bP[2d,4m-\ell-1][2d,4m-\ell-1]=\mathbf{0}$. It follows that the rank distance between them is at least $d$ and so $d_R(\bV[i,i+4m-1][j,j+4m-1],\bP)\geq d$. 

\vspace{2mm}
\noindent{\bf Case 2c}: $j\in [n-(4m-\ell)+1, n-1]$ and $i \in [0,4m-\ell-d]$. Then 
$$\bV[i,i+d-1][j+4m-d,j+4m-1]=\bV[i,i+d-1][j+4m-d-n,j+4m-1-n].$$
Since $0\leq i < i+d-1 \leq 4m-\ell-1$ and $\ell - d < j+4m-d-n < j+4m-1-n < 4m-1$, the subarray $\bV[i,i+d-1][j+4m-d,j+4m-1]$ is  an upper triangular matrix with all entries on the diagonal being $0$. Note that in $\bP$ the corresponding subarray $\bP[0,d-1][4m-d,4m-1]=\bI_d$. Thus the rank distance between $\bV[i,i+4m-1][j,j+4m-1]$ and $\bP$ is at least  $d$.

\vspace{2mm}
\noindent{\bf Case 2d}: $j\in [n-(4m-\ell)+1, n-1]$ and $i \in [4m-\ell-d+1, n-(4m-\ell)]$. Then $4m < i+\ell+d < i+4m-\ell-1 \leq n-1$. Since $4m-2\ell-d>3m$, the subarray $\bW[i+\ell+d, i+4m-\ell-1][j+\ell+d, j+4m-\ell-1]$ always contains an $m\times m$ submatrix of rank at least $d$. Note that the corresponding subarray $\bP[\ell+d, 4m-\ell-1][\ell+d, 4m-\ell-1]=\mathbf{0}$. It follows that $d_R(\bV[i,i+4m-1][j,j+4m-1],\bP)\geq d$.

\vspace{2mm}
\noindent{\bf Case 2e}: $j\in [n-(4m-\ell)+1, n-1]$ and $i \in [n-(4m-\ell)+1,n-1]$.  Then 
\begin{align*}
\bV&[i+4m-d,i+4m-1][j+4m-d,j+4m-1]\\
& =\bV[i+4m-d-n,i+4m-1-n][j+4m-d-n,j+4m-1-n].
\end{align*}
Since $\ell - d < i+4m-d-n <   j+4m-d-n$ and $i+4m-1-n < j+4m-1-n < 4m-1$, the subarray $\bV[i,i+d-1][j+4m-d,j+4m-1]$ is  a lower triangular matrix with all entries on the diagonal being $0$. Note that in $\bP$ the corresponding subarray $\bP[4m-d,4m-1][4m-d,4m-1]=\bI_d$. Thus the rank distance between  $\bV[i,i+4m-1][j,j+4m-1]$ and  $\bP$ is at least $d$.
\end{proof}

%We regard an array of size $(am) \times (bm)$ as a $a \times b$ partitioned matrix with each block being a square of side $m$. Given a pair of $(am) \times (bm)$ arrays $\bM_1$ and $\bM_2$, we denote the Hamming distance of their corresponding partitioned matrices as $d_{S}(\bM_1,\bM_2)$.  In other words, $d_{S}(\bM_1,\bM_2)$  counts the number of different squares in $\bM_1$ and $\bM_2$.  Therefore,  in Construction 3, %~\ref{Construction-RPA},  we have \[d_{S}(\bA_{ij},\bA_{i'j'})=d_H(\bc_{ij},\bc_{i'j'}).\] For a pair of $a\times b$ arrays $\bM_1$ and $\bM_2$ with $b$ not divisible by $m$,  we can repeat the last columns to form two arrays $\bM_1'$ and $\bM_2'$ of size  $a\times (\ceil{{b}/{m}}m)$.  Then denote $d_B(\bM_1,\bM_2):=d_B(\bM_1',\bM_2')$. Hence, $d_{B}(\bM_1,\bM_2)$ counts the number of different (truncated) squares in $\bM_1$ and $\bM_2$. 

\begin{lemma}
\label{Lem:ModularComparision-Array-r}
For any two subarrays $\bW=\bA[i,i+n_1-1][j,j+n_2-1]$ and $\bW'=\bA[i',i'+n_1-1][j',j'+n_2-1]$ with $i\equiv i'\pmod{n_1}$ and $j\equiv j'\pmod{n_2}$, the rank distance between them is at least $d$.
\end{lemma}

\begin{proof}
Suppose that $i=an_1+\bai$ and $i'=a'n_1+\bai$ for some $\bai \in \bbracket{n_1}$,  and $j=bn_2+\baj$ and $j'=b'n_2+\baj$ for some $\baj \in \bbracket{n_2}$. Let $\hai \in\bbracket{n_1}$ and $\haj\in \bbracket{n_2}$  be the integers such that $\bai+\hai\equiv 0\pmod{n_1}$ and $\baj+\haj\equiv 0\pmod{n_2}$. Shift $\bW$ cyclically upwards $\hai$ times and leftwards $\haj$ times and denote the resulting array as $\bV$. Similarly, let $\bV'$ be the corresponding shifted array of $\bW'$. Then $d_R(\bW,\bW')=d_R(\bV,\bV')$. 
To estimate $d_R(\bV,\bV')$, we proceed in three cases,   depending on where the check bits of $\bV$ and $\bV'$ come from. Similar to the proof of Lemma~\ref{Lem:ModularComparision-Array}, for any two subarrays $\bM$ and $\bM'$  in $\bA$ which are of same dimension and in the same modular position, we use $d_S(\bM,\bM')$ to denote the number of different (truncated) $m\times m$ subarrays in $\bM$ and $\bM'$.

\vspace{2mm}
\noindent{\bf Case 1}: $\bai\in [0, n_1-2m-1]$ and $\baj\in [0,n_2-4m-1]$. 
%For each $\bA_{\alpha \beta}$,
For $\alpha,\beta\in\bbracket{M}$, we change the bits in the subarrays $\bA_{\alpha \beta}[n_1-6m,n_1-4m-1][n_2-12m,n_2-8m-1]$ and $\bA_{\alpha \beta}[n_1-4m,n_1-2m-1][n_2-8m,n_2-4m-1]$    to one and 
denote the resulting array as $\bar{\bA}_{\alpha \beta}$. Let $\bar{\bc}_{\alpha \beta}$ be the corresponding shortened codeword of length $n_R-16$. 
Then we have 
\begin{align}
d_S(\bar{\bA}_{\alpha\beta},\bar{\bA}_{\alpha'\beta'})=d_H(\bar{\bc}_{\alpha\beta}, \bar{\bc}_{\alpha'\beta'})\geq (n_R-16)-k_R+1\geq 9.
\end{align}
Now, let $\bar{\bW}$, $\bar{\bW}'$, $\bar{\bV}$ and  $\bar{\bV}'$ be the corresponding arrays of ${\bW}$, ${\bW}'$, ${\bV}$ and  ${\bV}'$ with some check bits being changed to one. 
We  partition $\bar{\bW}$ and $\bar{\bA}_{ab}$ as in Fig.~\ref{Fig:ModularComparision-Array-r}(a). Then   
$$\bar{\bV}=
\left(\begin{array}{cc}
\bar{\bW}_\textup{IV} & \bar{\bW}_\textup{III}\\
\bar{\bW}_\textup{II} & \bar{\bW}_\textup{I}
\end{array}
\right),
\bar\bA_{ab}=
\left(\begin{array}{cc}
\bar{\bA}_\textup{IV} & \bar{\bA}_\textup{III}\\
\bar{\bA}_\textup{II} & \bar{\bA}_\textup{I}
\end{array}
\right),
\textup{\ and \ }
\bar{\bW}_\textup{I}=\bar{\bA}_\textup{I}.$$
%We propose to estimate $d_B(\bV,\bA_{ab})$. If $\baj$ is not divisible by $m$, we could repeat the last columns of $\bW_\textup{II}, \bW_\textup{IV},\bA_\textup{II}$, and $\bA_\textup{IV}$ or the first columns of $\bW_\textup{I}, \bW_\textup{III}, \bA_\textup{I}$, and $\bA_\textup{III}$, so that their widths are divisible by $m$.
Furthermore, we have 
\[ d_S(\bar{\bW}_\textup{II},\bar{\bA}_\textup{II}) \leq 1, \  d_S(\bar{\bW}_\textup{III},\bar{\bA}_\textup{III})\leq 1, \textup{\ and \ } d_S(\bar{\bW}_\textup{IV},\bar{\bA}_\textup{IV})\leq 2. \]
It follows that $d_S(\bar{\bV},\bar{\bA}_{ab}) \leq 4.$
With the same argument, we can get $d_B(\bar{\bV}',\bar{\bA}_{a'b'})\leq 4$. Hence,
$$d_S(\bV,\bV')\geq d_S(\bar{\bV},\bar{\bV}')\geq  d_S(\bar{\bA}_{ab},\bar{\bA}_{a'b'})-8\geq 1.$$
So we can find a pair of distinct $m\times m$ subarrays in the same position of  $\bV$ and $\bV'$. Since these two subarrays are codewords of an MRD code, the rank distance between them is at least $d$. Thus $d_R(\bW,\bW')=d_R(\bV,\bV')\geq d$.

\vspace{2mm}
\noindent{\bf Case 2}: $\bai\in [0, n_1-4m-1]$ and $\baj\in [n_2-4m, n_2-1]$. We change the bits in the subarrays $\bA_{\alpha \beta}[n_1-6m,n_1-4m-1][n_2-12m,n_2-8m-1]$ and $\bA_{\alpha \beta}[n_1-2m,n_1-1][n_2-4m,n_2-1]$    to one and 
denote the resulting array as $\tilde{\bA}_{\alpha \beta}$. Let $\tilde{\bW}$, $\tilde{\bW}'$, $\tilde{\bV}$ and  $\tilde{\bV}'$ be the corresponding arrays of ${\bW}$, ${\bW}'$, ${\bV}$ and  ${\bV}'$ with some check bits being changed to one.  Partition $\tilde{\bW}$ and $\tilde{\bA}_{a,b+1}$ as in Fig.~\ref{Fig:ModularComparision-Array-r}(b). Then   using the same strategy as in Case 1, we can show
\[d_S(\tilde{\bV},\tilde{\bA}_{ab})\leq 4, \ d_S(\tilde{\bV}',\tilde{\bA}_{ab}')\leq 4 \textup{\ and }  d_S(\tilde{\bV},\tilde{\bV}')\geq  d_S(\tilde{\bA}_{ab},\tilde{\bA}_{a'b'})-8\geq 1.\]
It follows that $d_R(\bW,\bW')=d_R(\bV,\bV')\geq d_R(\tilde{\bV},\tilde{\bV}')  \geq   d$.

\vspace{2mm}
\noindent{\bf Case 3}: $\bai\in [n_1-4m, n_1-1]$ and $\baj\in [n_2-4m, n_2-1]$. In the last case, we change the bits in the subarrays $\bA_{\alpha \beta}[n_1-4m,n_1-2m-1][n_2-8m,n_2-6m-1]$ and $\bA_{\alpha \beta}[n_1-2m,n_1-1][n_2-4m,n_2-1]$    to one and denote the resulting array as $\hat{\bA}_{\alpha \beta}$. Let $\hat{\bW}$, $\hat{\bW}'$, $\hat{\bV}$ and  $\hat{\bV}'$ be the corresponding arrays of ${\bW}$, ${\bW}'$, ${\bV}$ and  ${\bV}'$. Then  partition $\hat\bA_{a+1,b+1}$ as in Fig.~\ref{Fig:ModularComparision-Array-r}(c).
 \end{proof}

\begin{figure*}[!htb]
  \centering
  \includegraphics[width=8.5cm, trim={0 0 0 0}, clip]{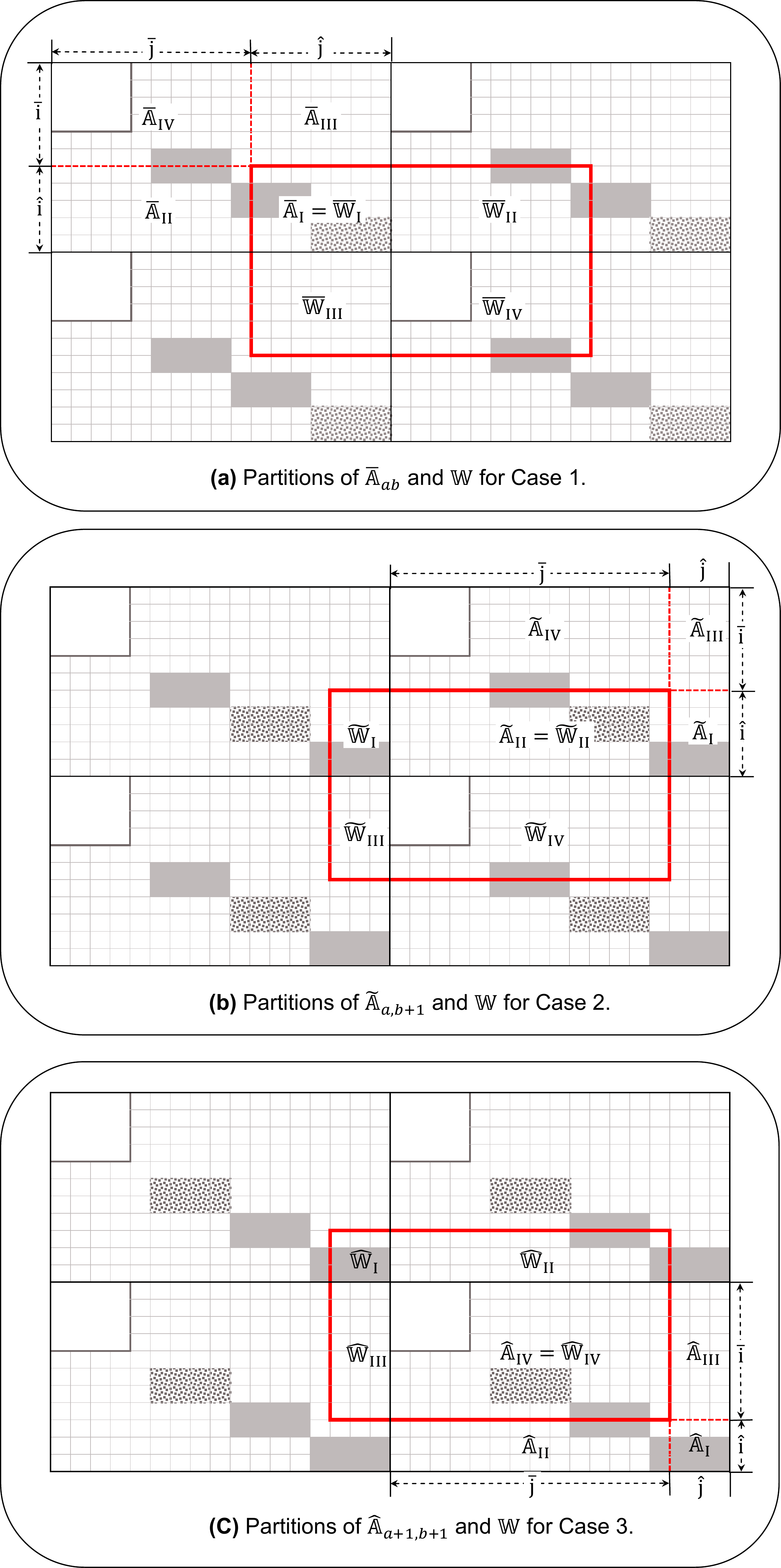}
  \caption{An example with $n_1=11m$ and $n_2=17m$ to illustrate the proof of Lemma~\ref{Lem:ModularComparision-Array-r}. The red lines enclose the subarray $\bW$. The empty $m\times m$ subarrays  represent  message bits. The subarrays with dots represent check bits. The solid subarrays represent the all-one matrices. }
  \label{Fig:ModularComparision-Array-r}
\end{figure*}

Using Lemma~\ref{Lem:Comparison-Array-r} and Lemma~\ref{Lem:ModularComparision-Array-r}, we have the following result.

\begin{theorem}
%The array $\bA$ in Construction~\ref{Construction-RPA} is an $(n_1\times n_2,d)$-RPA.
The array $\bA$ in Construction 3 is a positioning array in which the rank distance between any two $n_1\times n_2$ submatrices is at least $d$.
\end{theorem}

It can be checked that the redundancy of $\bA$ in Construction 3 is $n_2 (d-1) \frac{n_1}{m}+O(\log (n_1n_2))$, where $m(m-d+1)=\log(n_1n_2)$. In contrast, the Singleton bound suggests that  the redundancy is at least  $n_2(d-1)$.

\section{$q$-ary Robust Positioning Sequences}
\label{Sect:qaryRPS}
In this section,  we modify the construction of  Berkowitz and Kopparty and give a new class of $q$-ary positioning sequences robust to a constant fraction of errors.  We first review Berkowitz and Kopparty's work.

\begin{theorem}[Berkowitz and Kopparty \cite{BerkowitzKopparty:2016}]\label{BerKopcon} Fix  a generator $g$ of $\F_q^*$.  Let  $\cG=(\bsg_0,\bsg_1,\ldots, \bsg_{q^k-1})$ be a $(k,q)$-Gray code. For each $\bsg_i$, let $f_i(x)\in\F_q[x]$ be the unique interpolating polynomial of degree $k+1$
so that
\begin{enumerate}
\item  ${\rm coeff}_xf_i=1$, ${\rm coeff}_1f_i=0$;
\item $f_i(g^j)=\sigma[j]$ for all $0\leq j<k$.
\end{enumerate}
Define a sequence
$$\bct=\bt_0\bt_1\cdots \bt_{q^k-1},$$
where
$$\bt_i=(f_i(g^0),f_i(g^1), \ldots, f_i(g^{q-2})).$$
Then $\bct$ is an $(n,d)_q$-RPS with $n=q-1$ and $d\geq \max\left\{\frac{n-k}{3}-3, n-3k-9\right\}.$
\end{theorem}

\begin{corollary}[Berkowitz and Kopparty \cite{BerkowitzKopparty:2016}]\label{BerKopcor}
For any $0 < R < 1$ and $\delta < \max\left\{\frac{1-R}{3}, 1-3R\right\}$, for large enough $q$ there exists a $q$-ary robust positioning sequence of strength $n$, rate $R$ and relative
distance $\delta$.
\end{corollary}

%Berkowitz and Kopparty  used the following strategy to estimate  the distance of  subwords in $\bct$: break each subword into pieces, each of which is a rotation of a subword of a codeword from a cyclic  code, analyse what distance and agreement bounds one can get on each piece, and then recombine them to get the final estimation.

Now, we use a simple strategy to improve on the relative distance $\delta$: we map the symbols in some positions of $\bct$ to another alphabet which is disjoint with $\F_q$.  %and then the number of agreements on the corresponding  pieces become $0$.

\vspace{2mm}
\noindent{\bf Construction 4}. Let $E$ be a set of $q$ elements  which is disjoint from $\F_q$. Fix a one-to-one map $\chi$ from $\F_q$ to $E$.  For a vector $\bv=(v_0,v_1,\ldots, v_{\ell-1}) \in \F_q^\ell$, define $\chi(\bv)=(\chi(v_0), \chi(v_1), \chi(v_2), \ldots, \chi(v_{\ell-1}))$. Now, let $\bt_0,\bt_1,\ldots,\bt_{q^k-1}$ be the family of sequences defined in Theorem~\ref{BerKopcon}.
Construct two sequences
\[
\bct_a  =\ba_0\ba_1\cdots\ba_{q^k-1},  \textrm{ \ and \ } \bct_b  =\bb_0\bb_1\cdots\bb_{q^k-1},
\]
where
$$\ba_i= \bt_i \textrm{ if $i$   is even,  or} \  \ba_i= \chi(\bt_i) \textrm{ if $i$ is  odd;}  \ and$$
$$\bb_i=\bt_i[0,k-1]\chi(\bt_i[k,n-1]).$$
\vspace{2mm}

We have the following estimation on the distances of $\bct_a$. 

\begin{theorem}
The sequence $\bct_a$ in Construction~4 is an $(n,d)_{2q}$-RPS with $n=q-1$ and $d\geq n-2k$.
\end{theorem}

\begin{proof}
Let $\bw_1, \bw_2$ be two subwords of length $n$ in $\bct_a$, starting at positions $m_1$ and $m_2$ respectively. Let $m_1= in+\bam_1 \pmod{n}$ and $m_2= jn+\bam_2 \pmod{n}$, where $\bam_1,\bam_2 \in \bbracket{n}$.  Assume that $\bam_2\leq \bam_1$, then we can partition the interval $[0,n-1]$ into $3$ pieces by letting
\[
I_1 =[0,n-\bam_1-1], I_2 =[n-\bam_1,n-\bam_2-1],\textrm{ and } I_3 =[n-\bam_2,n-1].
\]
We consider the following cases.

{\bf Case 1.} First assume that both $i$ and $j$ are even, then $\bw_1[I_1]=\bt_i[\bam_1,n-1]$ and $\bw_2[I_1]=\bt_j[\bam_2, \bam_2+n-\bam_1-1]$. Noting that $\bt_i$ and $\bt_j$ are codewords of a  Reed-Solomon code, we have $\agree(\bw_1[I_1],\bw_2[I_1])\leq k$. Similarly,  $\agree(\bw_1[I_3],\bw_2[I_3])\leq k$.

Now for the interval $I_2$, we have $\bw_1[I_2]=\chi(\bt_{i+1}[0, \bam_1-\bam_2-1])$ and $\bw_2[I_2]=\bt_j[\bam_2+n-\bam_1,n-1]$, so the symbols of $\bw_1[I_2]$ come from $E$ and the symbols of $\bw_2[I_2]$ come  from $\F_q$. Since $E\cap \F_q=\emptyset$, $\agree(\bw_1[I_2], \bw_2[I_2])=0$. Hence $\agree(\bw_1, \bw_2)\leq 2k$ and $d_H(\bw_1,\bw_2)\geq n-2k$.

{\bf Case 2.} Here assume that $i$ is even and $j$ is odd, then  it is easy to see that the symbols of $\bw_1[I_1]$ and $\bw_2[I_3]$ are from $\F_q$ and the symbols of $\bw_1[I_3]$ and $\bw_2[I_1]$  are from $E$. Then $\agree(\bw_1[I_1],\bw_2[I_1])=\agree(\bw_1[I_3],\bw_2[I_3])=0$. Noting that $\bw_1[I_2]=\chi(\bt_{i+1}[0, \bam_1-\bam_2-1])$ and $\bw_2[I_2]=\chi(\bt_j[\bam_2+n-\bam_1,n-1])$, we have $\agree(\bw_1[I_2],\bw_2[I_2])\leq k$. Hence $d_H(\bw_1,\bw_2)\geq n-k$.

{\bf Case 3.}  The final case is when $i$ is odd. With the same argument as in Case 1 and Case 2,  we still can show that $d_H(\bw_1,\bw_2)\geq n-2k$.
\end{proof}

For the sequence $\bct_b$, we have the following result, the proof of which is similar  to that of \cite[Theorem~6]{BerkowitzKopparty:2016} and we omit here.

\begin{theorem}
The sequence $\bct_b$ in Construction~3 is an $(n,d)_{2q}$-RPS with $n=q-1$ and $d\geq \frac{n-k-9}{2}$.
\end{theorem}

\begin{corollary}
For any $0 < R < 1$ and $\delta < \max\left\{\frac{1-R}{2}, 1-2R\right\}$, for large enough $q$ there exists a $q$-ary robust positioning sequence of strength $n$, rate $R$ and relative
distance $\delta$.
\end{corollary}

\begin{proof}
$\bct_a$ and $\bct_b$ have the same rate:
\begin{align*}
R=&\frac{\log_{2q}(nq^k)}{n}=\frac{\log_q(nq^k)}{n\log_q(2q)}
=\frac{\log_q(nq^k)}{n} \frac{1}{1+\frac{1}{\log q}}\\
\geq & \left(\frac{k+1}{n}-o(1)\right)\left(1-O\left(\frac{1}{\log q}\right)\right)
= \frac{k+1}{n}-o(1).
\end{align*}
The relative distance of $\bct_a$ is $\delta_a\geq \frac{n-2k}{n}=1-2R-o(1)$, and  the relative distance of $\bct_b$ is $\delta_b\geq \frac{n-k-9}{2n}=\frac{1-R}{2}-o(1).$
 \end{proof}

Recall that the relative distance of $\bct$ constructed in Corollary~\ref{BerKopcor}
is less than $\max\left\{\frac{1-R}{3}, 1-3R\right\}$. So, the constructed arrays $\bct_a$ and $\bct_b$ have larger relative distance, i.e., $\max\left\{\frac{1-R}{2}, 1-2R\right\}$. In contrast, using the Singleton bound, it is easy to see that the relative distance should be no more that $1-R+o(1)$.

\section{The Maximum Length of a Binary Robust Positioning Sequence}
\label{Sect:optimalbinaryRPS}
In this section, we determine the exact value of $P(n,d)$ for $d\geq \floor*{2n/3}$. We require the following upper bound on $P(n,d)$.
\begin{proposition}[Plotkin Bound]\label{Plotkin-bound}
If $d$ is even and $2d>n$, then 
$P(n,d)\leq 2 \left\lfloor \frac{d}{2d-n}\right\rfloor+n-1;$
if $d$ is odd and $2d+1>n$, then 
$P(n,d)\leq 2 \left\lfloor \frac{d+1}{2d+1-n}\right\rfloor+n-1.$
\end{proposition}

\begin{theorem}\label{elarge}
If $\floor*{2n/3}+1 \leq d \leq n$, we have $P(n,d) = n+1$. 
\end{theorem}
\begin{proof}
According to the Plotkin bound, if $\floor*{2n/3}+1 \leq d \leq n$, we have $P(n,d) \leq n+1$. It is easy to see that the sequence 
$(01)^{\ceil*{n/2}}0^{n+1-2\ceil*{n/2}}$ is an $(n,d)$-RPS of length $n+1$.
\end{proof}

\begin{theorem}\label{0mod3}
Let  $n\equiv 0\pmod{3}$, then $P(n,2n/3) = n+2$.
\end{theorem}
\begin{proof} %We first consider the case of $n\equiv 0\pmod {3}$.
Let $\bs$ be an $(n,2n/3)$-RPS of length $N$. According to the Plotkin bound, we have that $N\leq n+3$. Suppose that $N=n+3$, then there are four subwords $\bc_1, \bc_2, \bc_3,$ and $\bc_4$, which are listed in a $4\times n$ matrix.
\begin{align*}
	\bc_1 &= x_1x_2\cdots x_n\\
	\bc_2 &= x_2x_3\cdots x_{n+1}\\
	\bc_3 &= x_3x_4\cdots x_{n+2}\\
	\bc_4 &= x_4x_5\cdots x_{n+3}
\end{align*}
Since the Plotkin bound is attained in this case,  the number of zeros and ones in each column of the matrix is equal. Hence, $x_{i}=x_{i+4}$, where $1 \leq i \leq n-1$. We consider the following cases.

{\bf Case 1.} If $x_1=x_2$, then $x_3=x_4$. It follows that $d_H(\bc_1,\bc_2) = \floor*{\frac{n}{2}} < \frac{2n}{3}$, which is a contradiction.

{\bf Case 2.} If $x_1=x_3$, then $x_2=x_4$ and hence $d_H(\bc_1,\bc_3) = 0 < \frac{2n}{3}$, which is a contradiction.

{\bf Case 3.} If $x_1=x_4$, then $x_2=x_3$ and hence $d_H(\bc_2,\bc_3) = \floor*{\frac{n}{2}} < \frac{2n}{3}$, which is a contradiction.

Therefore, we have that $N \leq n+2$. It is easy to see that the binary sequence $(100)^{n/3}(10)$ of length $n+2$ is an $(n,2n/3)$-RPS of length $n+2$. It follows that $P(n,2n/3)=n+2$.
%For $n\equiv 2\pmod{3}$,  the Plotkin bound shows that the length $N$ of an  $(n,(2n-1)/3)$-RPS is at most  $n+3$. If $N=n+3$, the Plotkin bound is absolutely  tight and we can use the same argument as above to show that such a sequence does not exist. On the other hand,  it is checked that the binary sequence $(100)^{(n+1)/3}1$ of length $n+2$ is an $(n,(2n-1)/3)$-RPS of length $n+2$. Thus, $P(n,(2n-1)/3)=n+2$ for $n\equiv 2\pmod{3}$.

\end{proof}

\begin{theorem}\label{1mod3}
Let $m$ be a positive integer, then we have that
  \begin{equation*}
\begin{split}
P(3m+1,2m) =  & \begin{cases}   7, \textup{\ \ when $m=1$;} \\
13, \textup{\ \ when $m=2$;} \\
14, \textup{\ \ when $m=3$;} \\
3m+4, \textup{\ \ when $m\geq 4$.}\\
\end{cases}
\end{split}
\end{equation*}  
\end{theorem}

\begin{proof}
An exhaustive search  shows that $P(4,2)=7$, $P(7,4)=13$ and $P(10,6)=14$. The corresponding optimal RPSs can be found in the Appendix. 

For $m\geq 4$, the Plotkin bound suggests that $P(3m+1,2m)\leq 3m+4$. In the followings, we give a recursive construction of RPSs with length  achieving this  bound. 

Let $$\bcs_1\triangleq 0001000.$$ For $m\geq 1$, let $$\bcs_{m+1}\triangleq \bcs_m[0,m+3] \overline{\bcs_m[m+4] }\bcs_m[m+2, 3m+3],$$
where  $\overline{\bcs_m[m+4]}$ is the complement of $\bcs_m[m+4]$. %In other words, $\bcs_{m+1}$ is obtained from $\bcs_m$ by inserting a short sequence $bcs_m[0,m+3] \overline{\bcs_m[m+4] }\bcs_m[m+2,$

Obviously, each $\bcs_m$ has length $3m+4$. By using  inductive arguments, one can see that for any $m\geq 1$, $$\bcs_m[m+1]=\bcs_m[m+4]$$ and $$(\bcs_m[m+2],\bcs_m[m+3],\overline{\bcs_m[m+4]})=(1,0,1) \textup{\ or } (0,1,0).$$  

Now, we use mathematical induction to show that the sequences constructed above are $(3m+1,2m)$-RPSs.
It is easy to check that $\bcs_1$ is a $(4,2)$-RPS. Assume that $\bcs_m$ is a $(3m+1,2m)$-RPS. Let $s_i$ be the $(1+i)$-th symbol of $\bcs_m$. Consider the following $4\times (3m+1)$ matrix. 
$$\left(
\begin{array}{ccccccc}
s_0 & \cdots & s_m &         s_{m+1} &         s_{m+2} & \cdots & s_{3m}\\
s_1 & \cdots & s_{m+1} &     s_{m+2} &      	 s_{m+3} & \cdots & s_{3m+1}\\
s_2 & \cdots & s_{m+2} &     s_{m+3} & 			s_{m+4} & \cdots & s_{3m+2}\\
s_3 & \cdots & s_{m+3} &       s_{m+4} &			 s_{m+5} & \cdots & s_{3m+3}\\
\end{array}
\right)$$
According to our assumption, any two rows of the matrix above have distance at least $2m$. Now, for $\bcs_{m+1}$, since   
$$\bcs_{m+1} = s_0 s_1 \cdots s_{m+1} s_{m+2} s_{m+3} \overline{s_{m+4}} s_{m+2} s_{m+3} \cdots s_{3m+3},$$
the four subwords of length $3m+4$  form the following matrix. 
$$\left(
\begin{array}{cccccccccc}
s_0 & \cdots & s_m &         s_{m+1} &         s_{m+2}   &   s_{m+3}   &    \overline{s_{m+4}}      &       s_{m+2} & \cdots & s_{3m}\\
s_1 & \cdots & s_{m+1} &     s_{m+2} &         s_{m+3}   &   \overline{s_{m+4}}   &    {s_{m+2}}      & 	 s_{m+3} & \cdots & s_{3m+1}\\
s_2 & \cdots & s_{m+2} &     s_{m+3} & 	\overline{s_{m+4}}   &   s_{m+2}   &    s_{m+3}      & 			s_{m+4} & \cdots & s_{3m+2}\\
s_3 & \cdots & s_{m+3} &       \overline{s_{m+4}} &	s_{m+2}   &   s_{m+3}   &    s_{m+4}      & 					 s_{m+5} & \cdots & s_{3m+3}\\
\end{array}
\right)$$
So the second matrix can be obtained from the first matrix by replacing the column $(s_{m+1}, s_{m+2}, s_{m+3}, s_{m+4})^T$ with 
$$\left(
\begin{array}{cccc}
   s_{m+1} &         s_{m+2}   &   s_{m+3}   &    \overline{s_{m+4}}    \\
    s_{m+2} &         s_{m+3}   &   \overline{s_{m+4}}   &    {s_{m+2}}   \\
    s_{m+3} & 	\overline{s_{m+4}}   &   s_{m+2}   &    s_{m+3}     \\
      \overline{s_{m+4}} &	s_{m+2}   &   s_{m+3}   &    s_{m+4}   \\
\end{array}
\right).$$
We look at  the first  row and the second row. Since $(s_{m+2}, s_{m+3}, \overline{s_{m+4}}) = (0,1,0)$ or $(1,0,1)$, the replacement increases the distance by two.  Similarly, for the other pairs of rows, except  the first and fourth ones, we can see that the distances are increased by two; for the first row and the fourth row, since $s_{m+1}=s_{m+4}$, the distance is increased again by two.  Thus, according to our assumption,  the distance between any two  rows in the second matrix is at least $2m+2$. The proof is completed. 
\end{proof}

Now, we look at the case of $n\equiv 2 \pmod{3}$.  A binary vector is called {\it balanced} if the number of ones and the number of zeros are equal.
Let $\bE_1$  and $\bE_2$ be the following infinite matrices. 
$$\bE_1\triangleq\left(
\begin{array}{cccccccc}
    \cdots & 0 &   0   &   1   &    1 & 0  & 0 & \cdots  \\
    \cdots & 0 &   1   &   1   &    0 & 0  & 1 & \cdots \\
    \cdots & 1 &   1   &   0   &    0 & 1  & 1 & \cdots     \\
    \cdots & 1 &	 0   &   0   &    1 & 1  & 0 & \cdots   \\
\end{array}
\right),$$
and 
$$\bE_2\triangleq\left(
\begin{array}{ccccccccc}
    \cdots & 0 &   1   &   0   &  1 & 0 &  1 & \cdots  \\
    \cdots & 1 &   0   &   1   &  0 & 1 & 0 & \cdots \\
    \cdots & 0 & 	 1   &   0   &  1 & 0 &  1 &\cdots     \\
    \cdots & 1 &	 0   &   1   &  0 & 1 & 0 & \cdots   \\
\end{array}
\right).$$

Let $\bE_1(\ell)$ be a $4\times \ell$ submatrix of $\bE_1$ with $\ell$ consecutive columns. Similarly, let $\bE_2(\ell)$ be a $4\times \ell$ contiguous submatrix of $\bE_2$.

\begin{theorem}\label{2mod3}
Let $m$ be a positive integer, then we have that
\begin{equation*}
P(3m+2,2m+1) =  3m + 4.
\end{equation*}  
\end{theorem}

\begin{proof}
Let $\bs$ be an $(3m+2,2m+1)$-RPS of length $N$. According to the Plotkin bound, we have that $N\leq 3m+5$. Suppose that $N=3m+5$, then there are four subwords of length $3m+2$, say, $\bc_1, \bc_2, \bc_3,$ and $\bc_4$, which can be listed in the following $4\times (3m+2)$ matrix.
\[
\begin{pmatrix}
    \bc_1\\
    \bc_2\\
    \bc_3\\
    \bc_4
\end{pmatrix}
= 
\begin{pmatrix} 
    x_1& x_2 & \cdots & x_{3m+2}\\
    x_2 & x_3 & \cdots & x_{3m+3}\\
    x_3 & x_4 & \cdots & x_{3m+4}\\
    x_4 & x_5 & \cdots & x_{3m+5}
\end{pmatrix} 
\]

We  consider the sum of the distances between $\bc_i$ and $\bc_j$, where $1\leq i \not= j \leq 4$. Since $d_H(\bc_i,\bc_j)\geq 2m+1$, this sum  is at least $24m+12$. In addition, every column contributes at most  8. So, the sum is at most $24m+16$. Therefore, there are at most two unbalanced columns and each of these unbalanced columns should has three identical symbols, i.e, it should be of one of the following forms:
\[
\begin{pmatrix}
    a\\
    b\\
    b\\
    b
\end{pmatrix},
\begin{pmatrix}
    b\\
    a\\
    b\\
    b
\end{pmatrix},
\begin{pmatrix}
    b\\
    b\\
    a\\
    b
\end{pmatrix}, \textup{\ and\ } 
\begin{pmatrix}
    b\\
    b\\
    b\\
    a
\end{pmatrix}, \text{\ where\ } a \not = b.
\]
Denote these forms as $\bu_1, \bu_2,\bu_2$ and $\bu_4$, respectively. 
We consider the following cases.

{\bf Case 1.} The sum of the distances is  $24m + 16$. Then all columns are balanced. The same argument as that in the proof of  Theorem~\ref{0mod3} leads to a contradiction.

{\bf Case 2.} The sum  is equal to $24m + 14$. There is only one unbalanced column. Hence, the matrix should be one of the following forms.
$$\bE_1(3m+1) \bu_1,   \  \bE_2(3m+1-x)  \bu_2 \bE_1(x),  \  \bE_1(x) \bu_3\bE_2(3m+1-x), \ \textup{\ or\ } \bu_4 \bE_1(3m+1).$$
If the matrix has form  $\bE_1(3m+1) \bu_1$, then $d_H(\bc_3,\bc_4) = \floor{(3m+1)/2} < 2m+1$, a contradiction. If the matrix  has form $\bE_2(3m+1-x)\bu_2\bE_1(x)$, we have that 
$d_H(\bc_1,\bc_4) = 3m+1-x +\floor{x/2}$ and  $d_H(\bc_1,\bc_3) = x$, both of which should be at least $2m+1$. That's impossible. For the other two cases, we may consider the reverse of $\bs$ to get the contradiction. 

{\bf Case 3.} The sum of the  distances is $24m + 12$. There are two unbalanced columns and we discuss in the following subcases, depending on the possible pair  of the unbalanced  columns.
\begin{enumerate}
\item If the pair has form $(\bu_1,\bu_4)$, then the matrix is of form $\bE_1(x)  \bu_1\bu_4 \bE_1(3m-x)$. In this case, we have $2m+1\leq d_H(\bc_1,\bc_2)\leq  3m/2+1$, a contradiction.
\item If the pair has form $(\bu_2,\bu)$ with  $\bu=(b,b,a,b)^T$ or $(a,a,b,a)^T$, then the matrix is of form  $\bE_2(x) \bu_2\bE_1(3m-x-z)\bu\bE_2(z)$. We have the following system. 
\begin{equation*}
\begin{split}
\begin{cases} 2m+1 \leq    d_H(\bc_1,\bc_4) \leq x +\ceil{(3m-x-z)/2} + z  \\
2m+1\leq d_H(\bc_1,\bc_3) = 3m-x-z + 1
\end{cases}
\end{split}
\end{equation*}
However, there are no solutions to this system.
\item If the pair has  form $(\bu_2,\bu)$ with  $\bu=(a,b,b,b)^T$ or $(b,a,a,a)^T$, then the matrix has form  $\bE_2(3m-x) \bu_2\bE_1(x) \bu$ and $x$ is even. In this case, we have the follow system, which has no solutions.  
\begin{equation*}
\begin{split}
\begin{cases} 2m+1\leq   d_H(\bc_1,\bc_3) = x+1 \\
2m+1\leq d_H(\bc_3,\bc_4) = 3m-x+x/2
\end{cases}
\end{split}
\end{equation*}
\item If the pair has  form $(\bu_3,\bu)$ with  with $\bu=(b,a,b,b)^T$ or $(a,b,a,a)^T$, then the matrix has form $\bE_1(x) \bu_3\bE_2(3m-x-z) \bu\bE_1(z)$. In this case, we have the following system which has no solutions. 
\begin{equation*}
\begin{split}
\begin{cases} 2m+1\leq   d_H(\bc_1,\bc_4) = \floor{x/2} + 3m-x-z + \floor{z/2} \\
2m+1\leq d_H(\bc_1,\bc_3) = x+z+1 
\end{cases}
\end{split}
\end{equation*}
\item If the pair has form $(\bu_4, \bu)$ with $\bu=(a,b,b,b)^T$ or $(b,a,a,a)^T$, then the matrix has form $ \bu_4 \bE_1(3m) \bu$ and we have that $2m+1\leq d_H(\bc_1,\bc_2)\leq  3m/2$, a contradiction. 
\item If the pair has form $(\bu_4, \bu)$ with  $\bu=(b,b,a,b)^T$ or $(a,a,b,a)^T$, then the matrix has form   $\bu_4\bE_1(x) \bu\bE_2(3m-x)$ and $x$ is even. In this case, we have the following system. 
\begin{equation*}
\begin{split}
\begin{cases} 2m+1\leq   d_H(\bc_1,\bc_2) = {x/2} + 3m-x \\
2m+1\leq d_H(\bc_1,\bc_3) = x+1
\end{cases}
\end{split}
\end{equation*}
However, there is no solutions to this system. 
\end{enumerate}

So far we have shown that $P(3m+2,2m+1)\leq 3m+4$. Note that $(100)^m (1001)$ is a $(3m+2,2m+1)$-RPS of length $3m+4$. The conclusion follows.
\end{proof}

By  exhaustive search, some values of $P(n,d)$ for  $2\leq d\leq n \leq 13$ are determined, see Table~\ref{tab.exact}. The corresponding optimal RPSs can be found in the Appendix. 

\begin{table}
\centering
\renewcommand{\arraystretch}{1}
\caption{Some values of $P(n,d)$ for $2 \leq d  \leq n \leq 13$}
\label{tab.exact}
\small
\begin{tabular}{cc|cccccccccccc}
\hline 
    & $d$ &  $2$ & $3$ & $4$ & $5$ & $6$ & $7$ & $8$ & $9$ & $10$ & $11$ & $12$ & $13$\\
 $n$ & & \\
 \hline
  $2$ & &   $3$ \\
  $3$ &  &  $5$ & $4$  \\    
  $4$ &  &  $7$ &  $5$ & $5$ \\
  $5$ &  &  $14$ & $7$ & $6$ & $6$ \\
  $6$ &  &  $23$ & $12$ & $8$ & $7$ & $7$ \\
  $7$ &  &  $41$ & $20$ & $13$ & $8$ & $8$ & $8$\\
  $8$ &  &  $74$ & $25$ & $15$ & $10$ & $9$ & $9$ & $9$ \\
  $9$ &  &  $\geq 137$ & $39$ & $19$ & $13$ & $11$ & $10$ & $10$ & $10$ \\
$10$ & &  $\geq 220$ & $71$ & $31$ & $20$ & $14$ & $11$ & $11$ & $11$ & $11$ \\
$11$ & & $\geq 324$  &  $\geq 137$ & $41$ & $32$ & $21$ & $13$ & $12$ & $12$ & $12$ & $12$ \\
$12$ &  & $\geq 598$ & $\geq 141$ & $73$ & $37$ & $23$ & $15$ & $14$ & $13$ & $13$ &   $13$ &  $13$ \\
$13$ &  &    &   &  & $43$  & $38$ & $19$ & $16$ & $14$ & $14$ & $14$ & $14$  & $14$ \\
 \hline
\end{tabular}
\end{table}

\section{Asymptotically Optimal Positioning Sequences of Distance $n-1$}
\label{Sect:optimalqaryRPS}

In this section, we study positioning sequences with large distance. Here, we require the Singleton bound.

\begin{proposition}[Singleton Bound]\label{Singleton-bound}
For all $n,d$ and $q$, we have that $P_q(n,d)\leq q^{n-d+1}+n-1$.
\end{proposition}
We aim to construct sequences with length close to the bound above.  Let $n=qt+s$ with $s\in \bbracket{q}$, it is easy to check that the sequence $(012\cdots(q-1))^{t+1}01\cdots(s-1)$ is an $(n,n)_q$-RPS of length $q+n-1$. So for all $n$ and $q$, we have that $$P_q(n,n)=q+n-1.$$

Now we turn to the case of $d=n-1$. We focus on cyclic sequences with the robust positioning ability. Formally, a cyclic sequence $\bs$ is a {\it $q$-ary cyclic positioning sequence} of  {\it strength} $n$ and  {\it distance} $d$ if for any $0\leq i<j<N$, $d_H(\bs[i,i+n-1],\bs[j,j+n-1])\geq d$. We  denote  such  a  sequence  as $(n,d)_q$-CRPS. The maximum length of an $(n,d)_q$-CRPS is denoted by $P_q^\circ(n,d)$.
Obviously, an $(n,d)_q$-CRPS is also an $(n,d)_q$-RPS;  furthermore, we can obtain a slightly longer $(n,d)_q$-RPS from the $(n,d)_q$-CRPS.

\begin{proposition}
$P_q(n,d)\geq P_q^\circ(n,d)+n-1$.
\end{proposition}

\begin{proof}
Let $\bs$ be an $(n,d)_q$-CRPS. Then the concatenation $\bs\bs[0,n-2]$ is an $(n,d)_q$-RPS.
\end{proof}

Let $\bs$ be a sequence of length $N$ over $\Sigma$. We say an ordered pair $(a,b)\in \Sigma^2$ {\it appears in $\bs$} if $\bs[i]=a$ and $\bs[j]=b$ for some $i,j\in \bbracket{N}$; furthermore, if $j\equiv i+\delta \pmod{N}$ for some $\delta \in \bbracket{N}$,  we say the pair $(a,b)$ {\it appears in $\bs$ with distance $\delta$}. We have the following  characterisation for the $(n,n-1)_q$-CRPS, the proof of which is straightforward and we omit here. 

\begin{proposition}\label{combnps} A sequence $\bs$ with alphabet $\Sigma$ is an  $(n,n-1)_{|\Sigma|}$-CRPS if and only if for each $\delta\in [1, n-1]$, every ordered pair $(a,b)\in \Sigma^2$ appears in $\bs$ with distance $\delta$ at most once.
\end{proposition}

We borrow the idea of \cite{Loncetal:2016} to construct  $(n,n-1)_q$-CRPSs.

\vspace{2mm}
\noindent{\bf Construction 5}.
Let $p$ and $r$ be two primes such that $p,r  > n$  and   $r^2\geq p-1$. For each $d\in \F_p\backslash \{0\}$, construct a sequence $\bc_d$ over $\F_p$ as
$$\bc_d=(d,2d,\ldots,(p-1)d).$$
Denote $E=\bbracket{n}\times \F_r$. For each $(a,b) \in \F_r^2$, construct a sequence $\bs_{a,b}$ over $E$ as
$$\bs_{a,b}=((0,b), (1,a+b),\ldots, (n-1,(n-1)a+b)).$$
Since $r^2\geq p-1$, take $p-1$ sequences from the collection of $\bs_{a,b}$'s and relabel them as $\bs_1, \bs_2, \ldots, \bs_{p-1}$.

Let $\Sigma=\F_p\cup E$ and construct a sequence $\bcc$ over $\Sigma$ as  $$\bcc=\bc_1\bs_1\bc_2\bs_2\cdots\bc_{p-1}\bs_{p-1}.$$

\begin{theorem}\label{checkseq}
The sequence $\bcc$ is an $(n,n-1)_q$-CRPS with $q=p+nr$.
\end{theorem}

\begin{proof}
According to Proposition~\ref{combnps}, we only need to show that for any $1\leq \delta \leq n-1$, every ordered pair in $\Sigma^2$ appears in $\bcc$ with distance $\delta$ at most once.

{\bf Case 1.} We first consider the pairs in $\F_p^2$. Suppose that the pair $(\alpha, \beta)\in \F_p^2$ appears in $\bcc$  with distance $\delta$ twice.  Since each $\bs_{a,b}$ is a length-$n$ sequence over $E$ and $E \cap \F_p=\emptyset$,   $(\alpha, \beta)$ must appear in the subsequences $\bc_d$ and $\bc_{d'}$ for some $d,d'\in[1,p-1]$. Then we may assume that $id=i'd'=\alpha$ and $(i+\delta)d=(i'+\delta)d'=\beta$, where $i,i'\in [1,p-1]$ and $i\not=i'$ if $d=d'$. It follows that $\beta-\alpha=\delta d=\delta d'$. Noting that $0< \delta <n <p$, we have $d=d'$. This in turn implies that $i=i'$ since $id=i'd'$, a contradiction. Thus for each $1\leq \delta \leq n-1$, every pair in $\F_p^2$ appears in $\bcc$  with distance $\delta$ at most once. 

{\bf Case 2.} Then we  consider the pairs in $E^2$. Suppose that the pair $((i, \alpha), (j,\beta))\in E^2$ appears in $\bcc$  with distance $\delta$ twice.  Since each $\bc_{d}$ is a length-$(p-1)$ sequence over $\F_p$ with $p-1\geq n$ and $E\cap F_p=\emptyset$,  we may assume that $j=i+\delta$, $ia+b=ia'+b'=\alpha$ and $ja+b=ja'+b'=\beta$, where  $a,a', b,b'\in \F_r$ and $(a,b)\not=(a',b')$.  It follows that $\beta-\alpha=\delta a=\delta a'$ and then $a=a'$. This leads to $b=b'$ as $ia+b=ia'+b'=\alpha$, a contradiction. Thus for each $1\leq \delta \leq n-1$, every pair in $E^2$ appears in $\bcc$  with distance $\delta$ at most once. 

{\bf Case 3.} For the pair $(\alpha, (i,\beta))$ in $F_p\times E$, suppose that it appears in $\bcc$ with distance $\delta$ twice for some $\delta < n$, then it appears in the subsequences $\bc_d\bs_d$ and $\bc_{d'}\bs_{d'}$ for some $d,d'\in [1,p-1]$. Since the symbol $(i,\beta)$  can only appear at the $i$-th position of $\bs_d$ and $\bs_{d'}$, the two sequences $\bc_d$ and $\bc_{d'}$ must have the same symbol $\alpha$ at some position $\ell$, i.e. $\ell d= \ell d'$, for some $\ell \in [1,p-1]$. It follows that $d=d'$, a contradiction. 

{\bf Case 4.} For the pair $((i,\alpha),\beta)$ in $E\times F_p$, with the same argument as in Case 3, we can show that it appears in $\bcc$ with distance $\delta$ at most once. 
\end{proof}

\begin{corollary}
Let $n=\lfloor  cq^\alpha \rfloor$, where $c$ and $\alpha$ are real numbers such that $c>0$ and $0\leq \alpha <\frac{1}{2}$. Then

\begin{equation*}
P_q^\circ(n,n-1) \geq q^2 - o(q^2).
\end{equation*}
\end{corollary}

\begin{proof}
Set $x = \left(\sqrt{q+n^2}-n\right)^2$. Then $x=q-2n\left(\sqrt{q+n^2}-n\right)=q-O\left(q^{\frac{1}{2}+\alpha}\right)$. According to Lemma~\ref{prime}, for sufficient large $q$,  we can choose a prime $p$ such that $x-x^\theta \leq  p\leq x$. Furthermore, according to   Bertrand's postulate,   we can choose a prime $r$ such that $p-1\leq  r^2\leq 4p$. Then $p=\Theta(q)$ and $r=\Theta(p^{1/2})=\Theta(q^{1/2})$. So we have $p, r^2 > n$ and we can apply Construction~5 to obtain an $(n,n-1)_{p+nr}$-CRPS, where

\begin{equation}\label{eqation}
p+nr\leq p+2n\sqrt{p}\leq  \left(\sqrt{q+n^2}-n\right)^2+2n\left(\sqrt{q+n^2}-n\right)= q.
\end{equation}

The length of $\bcc$ is
\begin{align*}
(p-1)(p-1+n) &\geq p^2-2p= (x-x^\theta)^2-2x  \geq x^2-O(x^{1+\theta}) \\
  & \geq q^2-O(q^{\frac{3}{2}+\alpha}) -O(q^{1+\theta})\geq q^2-o(q^2).
\end{align*}

\end{proof}

\begin{corollary}
Let $q=\Omega(n^{2+\epsilon})$ for some $\epsilon>0$. Then

\begin{equation*}
q^2 +n-1- o(q^2) \leq P_q(n,n-1) \leq q^2+n-1.
\end{equation*}
\end{corollary}

\section{Conclusion}
We construct binary positioning patterns, equipped with efficient locating algorithms, that are robust to a constant number of errors. Our strategy is based on $d$-auto-cyclic vectors, Reed-Solomon codes and Gray codes, and we reduce the number of redundancies as compared to previous constructions. In the locating algorithms, the $d$-auto-cyclic vectors are used as markers to locate the relative position. This information, together with the property of Gray codes, allows one to leverage the well-known fast decoding of Reed-Solomon codes to quickly identify the location.

\section*{Appendix}
Here we list optimal $(n,d)$-RPSs for  $n\leq 13$ and $2\leq d < \floor*{2n/3}$, as well as $(n,d)\in\{(4,2),(7,4),(10,6)\}$.
{\small

$(4,2)$-RPS:

$0001000$

$(5,2)$-RPS:

$00010111010001$

$(6,2)$-RPS:

$01001110010000101101010$

$(6,3)$-RPS:

$000101100010$

$(7,2)$-RPS:

$00101001101011001000111011100001011111010$

$(7,3)$-RPS:

$00001001111011000010$

$(7,4)$-RPS:

$0001011000101$

$(8,2)$-RPS:

$00100101010010011001000000100011010001011110101110110111000011100111110010$

$(8,3)$-RPS:

$0001000111101110100101000$

$(8,4)$-RPS:

$000010110000101$

$(9,3)$-RPS:

$000001000111010100101111001101100000100$

$(9,4)$-RPS:

$0001001011100010010$

$(9,5)$-RPS:

$0001101001110$

$(10,3)$-RPS:

$00000010001101101011001111000101111110111001001010011000011101000000100$

$(10,4)$-RPS:

$0000100100011110110111000010010$

$(10,5)$-RPS:

$00010010111000100101$

$(10,6)$-RPS:

$00011010110001$

$(11,4)$-RPS:

$00000100110001111001010110111010000010011$

$(11,5)$-RPS:

$00001001000111101101110000100100$

$(11,6)$-RPS:

$000100101110001001011$

$(12,4)$-RPS:

$0000001000110110101100111100010111111011100100101001100001110100000010001$

$(12,5)$-RPS:

$0000010101100111110101001100000101011$

$(12,6)$-RPS:

$00000110101100000110101$

$(12,7)$-RPS:

$000101001100111$

$(13,5)$-RPS:

$0000010011000111100101011011101000001001100$

$(13,6)$-RPS:

$00000101011001111101010011000001010110$

$(13,7)$-RPS:

$0001011000101100010$

}


\begin{thebibliography}{10}

\bibitem{Bakeretal:2001}
{\sc R.~C. Baker, G.~Harman, and J.~Pintz}, {\em The difference between
  consecutive primes, {II}}, Proc. London Math. Society, 83 (2001),
  pp.~532--562.

\bibitem{BerkowitzKopparty:2016}
{\sc R.~Berkowitz and S.~Kopparty}, {\em Robust positioning patterns}, in Proc.
  ACM-SIAM Symp. Discrete Algorithms, Society for Industrial and Applied
  Mathematics, 2016, pp.~1937--1951.

\bibitem{bruckstein2012simple}
{\sc A.~M. Bruckstein, T.~Etzion, R.~Giryes, N.~Gordon, R.~J. Holt, and
  D.~Shuldiner}, {\em Simple and robust binary self-location patterns}, IEEE
  Trans. Inform. Theory, 58 (2012), pp.~4884--4889.

\bibitem{Daietal:2014}
{\sc J.~Dai and C.-K.~R. Chung}, {\em Touchscreen everywhere: on transferring a
  normal planar surface to a touch-sensitive display}, IEEE Trans. Cybernetics,
  44 (2014), pp.~1383--1396.

\bibitem{etzion1988constructions}
{\sc T.~Etzion}, {\em Constructions for perfect maps and pseudorandom arrays},
  IEEE Trans. Inform. Theory, 34 (1988), pp.~1308--1316.

\bibitem{Gabidulin:1985}
{\sc E.~M. Gabidulin}, {\em Theory of codes with maximum rank distance},
  Problemy Peredachi Informatsii, 21 (1985), pp.~3--16.

\bibitem{Geng:2011}
{\sc J.~Geng}, {\em Structured-light {3D} surface imaging: a tutorial},
  Advances in Optics and Photonics, 3 (2011), pp.~128--160.

\bibitem{Guan:1998}
{\sc D.-J. Guan}, {\em Generalized gray codes with applications}, in Proc.
  Natl. Sci. Counc. Repub. China Part A Phys. Sci. Eng., vol.~22, 1998,
  pp.~841--848.

\bibitem{Hagitaetal:2008}
{\sc M.~Hagita, M.~Matsumoto, F.~Natsu, and Y.~Ohtsuka}, {\em Error correcting
  sequence and projective de {B}ruijn graph}, Graphs and Combinatorics, 24
  (2008), pp.~185--194.

\bibitem{KumarWei:1992}
{\sc P.~V. Kumar and V.~K. Wei}, {\em Minimum distance of logarithmic and
  fractional partial $m$-sequences}, IEEE Trans. Inform. Theory, 38 (1992),
  pp.~1474--1482.

\bibitem{LeviYaakobi:2017}
{\sc M.~Levy and E.~Yaakobi}, {\em Mutually uncorrelated codes for {DNA}
  storage}, IEEE Trans. Inform. Theory,  (2018),
  \url{https://doi.org/10.1109/TIT.2018.2873138}.

\bibitem{Loncetal:2016}
{\sc Z.~Lonc and M.~Truszczy{\'n}ski}, {\em Packing analogue of $k$-radius
  sequences}, European Journal of Combinatorics, 57 (2016), pp.~57--70.

\bibitem{macwilliams1976pseudo}
{\sc F.~J. MacWilliams and N.~J. Sloane}, {\em Pseudo-random sequences and
  arrays}, Proc. IEEE, 64 (1976), pp.~1715--1729.

\bibitem{Mitchelletal:1996}
{\sc C.~J. Mitchell, T.~Etzion, and K.~G. Paterson}, {\em A method for
  constructing decodable de {B}ruijn sequences}, IEEE Trans. Inform. Theory, 42
  (1996), pp.~1472--1478.

\bibitem{paterson1994perfect}
{\sc K.~G. Paterson}, {\em Perfect maps}, IEEE Trans. Inform. Theory, 40
  (1994), pp.~743--753.

\bibitem{roth1991maximum}
{\sc R.~M. Roth}, {\em Maximum-rank array codes and their application to
  crisscross error correction}, IEEE Trans. Inform. Theory, 37 (1991),
  pp.~328--336.

\bibitem{Sheinerman:2001}
{\sc E.~R. Scheinerman}, {\em Determining planar location via complement-free
  de {B}rujin sequences using discrete optical sensors}, IEEE Trans. Robotics
  and Automation, 17 (2001), pp.~883--889.

\bibitem{Stylus}
{\sc T.~Simonite}, {\em Microsoft mulls a stylus for any screen}, 2012,
  \url{http://www.technologyreview.com/news/428521/microsoftmulls-a-stylus-for-any-screen/}.

\bibitem{Szentandrasi:2012}
{\sc I.~Szentandrasi, M.~Zachari{\'a}{\v{s}}, J.~Havel, A.~Herout, M.~Dubska,
  and R.~Kajan}, {\em Uniform marker fields: Camera localization by orientable
  de bruijn tori}, in Proc. IEEE Intl. Symp. Mixed and Augmented Reality, 2012,
  pp.~319--320.

\bibitem{WelchBerlekamp:1986}
{\sc L.~R. Welch and E.~R. Berlekamp}, {\em Error correction for algebraic
  block codes}, Dec.~30 1986.
\newblock US Patent 4,633,470.

\bibitem{Yazdietal:2018}
{\sc S.~H.~T. Yazdi, H.~M. Kiah, R.~Gabrys, and O.~Milenkovic}, {\em Mutually
  uncorrelated primers for {DNA}-based data storage}, IEEE Trans. Inform.
  Theory, 66 (2018), pp.~6283--6296.

\end{thebibliography}
\end{document}